\documentclass[12pt, draftclsnofoot, onecolumn]{IEEEtran}
\usepackage{subcaption}
\usepackage{cite}
\usepackage{mathtools}
\usepackage{amssymb}
\usepackage{booktabs}
\usepackage{comment}

\ifCLASSINFOpdf
   \usepackage[pdftex]{graphicx}
  \usepackage{xcolor}

   \graphicspath{Fig/}

\else

  \usepackage{here}
  \usepackage[dvips]{graphicx}
  \usepackage[driver]{graphicx}
  \usepackage{epstopdf}
  \AppendGraphicsExtensions{.tif}
  \graphicspath{Fig/}

\fi
\usepackage{graphicx}

\usepackage{tabularx}
\usepackage{amsmath}
\usepackage{amsfonts}
\usepackage{amsthm}
\usepackage{bm}
\usepackage{mathtools}

\DeclarePairedDelimiter\floor{\lfloor}{\rfloor}

\usepackage{algorithm,algpseudocode}

\usepackage{enumitem}

\newtheorem{thm}{Theorem}
\newtheorem{lem}{Lemma}

\newtheorem{cor}{Corollary}

\theoremstyle{definition}
\newtheorem{defn}{Definition}

\newtheorem{con}{Condition}

\begin{document}

\title{Geometric Sequence Decomposition with $k$-simplexes Transform}

\author{Woong-Hee Lee, Jong-Ho Lee,~\IEEEmembership{Member,~IEEE}, and Ki Won Sung,~\IEEEmembership{Member,~IEEE}
\thanks{Woong-Hee Lee and Ki Won Sung are with the School of Electrical Engineering and Computer Science, KTH Royal Institute of Technology, Kista, 164 40, Sweden (e-mail:woolee@kth.se; sungkw@kth.se).}
\thanks{Jong-Ho Lee is with the School of Electronic Engineering, Soongsil University, Seoul 06978, South Korea (e-mail:jongho.lee@ssu.ac.kr).}
}


\maketitle

\begin{abstract}
This paper presents a computationally efficient technique for decomposing non-orthogonally superposed $k$ geometric sequences.
The method, which is named as geometric sequence decomposition with $k$-simplexes transform (GSD-ST), is based on the concept of transforming an observed sequence to multiple $k$-simplexes in a virtual $k$-dimensional space and correlating the volumes of the transformed simplexes.
Hence, GSD-ST turns the problem of decomposing $k$ geometric sequences into one of solving a $k$-th order polynomial equation.
Our technique has significance for wireless communications because sampled points of a radio wave comprise a geometric sequence.
This implies that GSD-ST is capable of demodulating randomly combined radio waves, thereby eliminating the effect of interference.
To exemplify the potential of GSD-ST, we propose a new radio access scheme, namely non-orthogonal interference-free radio access (No-INFRA). Herein, GSD-ST enables the collision-free reception of uncoordinated access requests.
Numerical results show that No-INFRA effectively resolves the colliding access requests when the interference is dominant.
\end{abstract}

\IEEEpeerreviewmaketitle

\begin{IEEEkeywords}
geometric sequence decomposition,
$k$-simplexes transform,
non-orthogonal interference-free radio access
\end{IEEEkeywords}

\section{Introduction}

\IEEEPARstart{A}{} geometric sequence is a series of numbers in which the ratio between any two consecutive terms is fixed. Recall that a geometric sequence is expressed by
\begin{equation*}
    \{a, ar, ar^2, ar^3, \cdots \}
\end{equation*}
where $a$ is the initial term and $r$ is the common ratio of the sequence. Depending on $r$, the geometric sequence can increase, decrease, or remain constant as it progresses. The sequence may also oscillate in the complex plane if the common ratio is a complex number.

Consider $k$ geometric sequences with nonidentical common ratios, $\mathbf{s}_1, \mathbf{s}_2, \cdots, \mathbf{s}_k$. Assume that we have no information on the individual sequences and can observe only a superposition of $k$ geometric sequences $\mathbf{s}$:
\begin{equation} \label{superposition}
    \mathbf{s} = \mathbf{s}_1 + \mathbf{s}_2 + \cdots + \mathbf{s}_k = \left\{ \sum_{n=1}^{k} a_n,  \sum_{n=1}^{k} a_n r_n, \sum_{n=1}^{k} a_n r_n^2, \cdots \right\}.
\end{equation}
Let us pose the question as follows:
\begin{itemize}
  \item \emph{How can we decompose a superposition of geometric sequences into the individual sequences in a computationally efficient manner?}
\end{itemize}
To answer the question, $k$ should first be determined. Second, the parameters of each sequence, i.e., $a_i$ and $r_i$, should be determined.

The foremost contribution of this paper is to propose a new technique addressing the aforementioned problem. The main idea is to transform the observed sequence to a $k$-dimensional space using a well-known concept in geometry: $k$-simplex \cite{grunbaum1969convex}. We develop a method using this transform and call it geometric sequence decomposition with $k$-simplexes transform (GSD-ST). Our method turns the complicated problem of decomposing $k$ geometric sequences into a simple root-finding for a $k$-th order polynomial equation. GSD-ST requires only $2k+1$ samples of the superposed sequence to obtain $k$ and retrieve the parameters of each sequence. The number of required samples reduces further to $2k$ if $k$ is known a priori.

The proposed GSD-ST is a noteworthy mathematical tool and has significance for wireless communications. This is because a sampling of a radio wave is a geometric sequence. A radio wave is generally represented by a complex-valued function of the form $A e^{i2 \pi ft}$. Here, $A$ is a constant that accounts for the amplitude and phase, and $f$ and $t$ are the frequency and time, respectively. It is observed that a sampled progression of a radio wave with time interval $\Delta t$ forms a geometric sequence with the initial term $A$ and common ratio $e^{i 2\pi f \Delta t}$. This implies that if we can decompose a superposition of geometric sequences, we can also separate multiple incoming radio waves that are non-orthogonally accumulated. Therefore, our work lays the foundation for new methods of handling wireless signals in interference-limited environments.

\subsection{Related Works}

Interference is a phenomenon wherein multiple waves superpose to form a resultant wave. Because it is common in various fields dealing with waves, extensive studies have been conducted to extract a desired wave or to separate all the accumulated waves. These include studies on blind sound source separation in acoustics \cite{pham2003blind,sawada2019similarity}, target detection in radar systems \cite{wang2014radar,xu2017information}, and wireless communications \cite{kwan2009proportional, trivedi2014comparison, pang2004performance, sydor2014cognitive, zhu2004survey, perahia2013next, ghosh2010lte,dai2015non, ding2015cooperative, islam2016power, choi2017noma, lee2018beamforming}.

In wireless communication systems, interference has traditionally been managed by orthogonalizing the signals. One of the most intuitive approaches to orthogonality is to coordinate the multiple waves in the time domain. Therefore, time-division multiplexing (TDM) remains as the fundamental principle for user scheduling in cellular networks \cite{kwan2009proportional, trivedi2014comparison} and collision avoidance in Wi-Fi systems \cite{pang2004performance, sydor2014cognitive}. In the frequency domain, orthogonal frequency-division multiplexing (OFDM) has become a key element of modern broadband systems such as the IEEE 802.11 family \cite{zhu2004survey, perahia2013next} and Long Term Evolution (LTE) \cite{ghosh2010lte}.

The scarcity of radio spectrum compelled the researchers to advance beyond the orthogonal division of radio resources. Furthermore, the stringent requirements of high data rate and low latency in 5G magnify the need. Thus, numerous attempts have been undertaken to address the non-orthogonal accumulation of radio waves. If we confine the discussion to radio access where the challenge is to accommodate multiple uncoordinated requests with ultra-low latency, non-orthogonal multiple access (NOMA) has been considered to be a practical solution \cite{dai2015non, ding2015cooperative, islam2016power, choi2017noma, lee2018beamforming}. This technique separates the multiple signals in the power domain through iterative decoding, i.e., it relies on the power difference between received signals. Therefore, its effectiveness reduces as the number of accumulated signals increases and the power difference decreases. A new method that can separate several randomly superposed radio waves regardless of the distribution of the received powers would be highly effective for designing radio access schemes for ultra-low latency. We demonstrate that the proposed GSD-ST is a strong candidate.

\subsection{Main Contributions}
Our objective is to decompose a non-orthogonal superposition of geometric sequences into the original sequences without information loss when we can observe only the superposed sequence. We propose GSD-ST, which is a computationally efficient method, for achieving this. We provide the fundamental concept, simple numerical examples, and the formal methodology of GSD-ST in the subsequent sections. We also propose a practical de-noising technique for GSD-ST because the method may be prone to the effect of noisy observations.

As discussed earlier, the decomposition of geometric sequences is equivalent to the separation of non-orthogonally overlapping radio waves. We introduce a new radio access scheme to illustrate GSD-ST's potential for wireless communications. Named as non-orthogonal interference-free radio access (No-INFRA), it enables multiple transmitters to randomly select frequencies in a continuous domain within a specified signal bandwidth and to transmit simultaneously. Then, the receiver samples the mixed signal and decodes the information using GSD-ST. Unlike orthogonal resource division which inevitably suffers from collisions of access requests, No-INFRA achieves collision-free access by eliminating the notion of interference.

To summarize, our main contributions in this study are threefold: first, we propose a new method (GSD-ST) for decomposing non-orthogonally superposed geometric sequences; second, we provide the formal methodology consisting of theorems and proofs to sustain the GSD-ST method; third, we introduce a new radio access scheme (No-INFRA) to demonstrate how GSD-ST can be applied to wireless communications.

\subsection{Organization of the Paper}
The remainder of this paper is organized as follows:
In Section II, we explain the fundamental concept of GSD-ST and provide numerical examples to enable the readers to understand the new method. Section III presents the formal methodology of GSD-ST, the necessary theorems, and a practical de-noising process. In Section IV, we propose the No-INFRA scheme. In Section V, the performance of No-INFRA is compared with a conventional scheme of orthogonal resource division. Finally, Section VI presents the concluding remarks.

\subsection{Notations}

The following symbols are used throughout the paper:

\begin{itemize}
\item  {$\mathbb{N}_0$: $\{0\}\cup\mathbb{N}$}.
\item $\mathbf{b}:= \{ \mathbf{b}[l] \}_{l=0}^{P-1} \in \mathbb{C}^{P}$: an arbitrary sequence of length $P$ whose $l$-th element is $\mathbf{b}[l]$.
\item $k$: the number of superposed geometric sequences.
\item  {$\hat{k}$: an estimate of $k$.}
\item $\mathbf{s}_n$: the $n$-th geometric sequence.
\item $a_n\in \mathbb{C}$: the initial term of $\mathbf{s}_n$, ($\mathbf{a}:=\{a_1,\cdots, a_k\}$).
\item $r_n\in \mathbb{C}$: the common ratio of $\mathbf{s}_n$, ($\mathbf{r}:=\{r_1,\cdots, r_k\}$).
\item $\mathbf{s}(:=\sum_{n=1}^{k} \mathbf{s}_n)$: {the superposed sequence of the geometric sequences.}
\item  {$\text{card}(\cdot)$: the cardinality of a collection.}
\item {$(\cdot)^{\text{T}}$: the transpose of a matrix.}
\item {$\det(\cdot)$: the determinant of a square matrix.}
\item $e(v_0 , \cdots , v_{k-1})$: a function that returns the $k$-simplex by connecting the $k+1$ $k$-vertices which consists of the origin and the specified $k$ $k$-vertices, $v_0 , \dots , v_{k-1}$, in a $k$-dimensional space.
\item $\Lambda$: a function that returns the volume of the $k$-simplex \cite{stein1966note}, i.e., $\Lambda(e(v_0 , \dots , v_{k-1})):= \frac{\det([v_0 , \dots , v_{k-1}])}{k!}$. In addition, if the input is a series of $k$-simplexes, this function returns the series of the volume of each $k$-simplex as the output.
\item $\phi_{k}$ ($\in  {\mathbb{N}_0^{k}}$): an arbitrary collection of lexicographically ordered $k$ indices, e.g., $\phi_{3}=\{ 0 , 2, 7\}$.
\item $\phi_{{k},\mathbf{b}} := ( \mathbf{b}[\phi_k [0]],\cdots, \mathbf{b}[\phi_k [k-1]] )^{\text{T}} \in \mathbb{C}^{k }$: the $k$-vertex in a $k$-dimensional space which is sketched by $\phi_k$ and made by $k$ samples of the sequence $\mathbf{b}$.
\item $\mathbf{1}_k$: the one-vector whose length is $k$.
\end{itemize}

\section{Overview of GSD-ST}

\subsection{Fundamental Concept of GSD-ST}

Recalling \eqref{superposition}, our problem is to decompose a superposition of $k$ geometric sequences, $\mathbf{s}$, when we have no information on the individual sequences and can observe only $\mathbf{s}$. The problem has $2k+1$ unknowns, i.e., the unknowns with regard to $\mathbf{s}_n$ ($a_n$ and $r_n$) and $k$. Thus, in principle, $2k+1$ observations of $\mathbf{s}$ would be sufficient to solve the problem, i.e., to obtain $k$, $\mathbf{a}$, and $\mathbf{r}$. To our knowledge, there is no computationally efficient method of solving the problem.

The superposition is a one-dimensional progression, but it contains the information on the $k$ sequences. Therefore, we depart from the intuition that an appropriate transformation of the observed sequence to a $k$-dimensional space may facilitate the analysis of the overlap of $k$ geometric sequences. To achieve this, we employ the concept of $k$-simplex in geometry because it represents the simplest possible polytope in $k$-dimensional space \cite{grunbaum1969convex}. $k$-simplex is defined as a $k$-dimensional polytope which is the convex hull of its $k+1$ $k$-vertices. For example, a $3$-simplex means a tetrahedron in a three-dimensional space.

We elaborate on our intuition to obtain $k$: if we create a series of $k$-simplexes from $\mathbf{s}$, it may hold a particular relationship. We observed that the volumes of the $k$-simplexes constitute a new geometric sequence only when $k$ is assumed correctly. Next, we consider $\mathbf{r}$, which are the nonlinear parameters of $\mathbf{s}$. For a single geometric sequence, e.g., $\mathbf{s}_1$, we can conveniently obtain $r_1$ from any two consecutive samples of $\mathbf{s}_1$. That is, 1-simplex is sufficient for analyzing a geometric sequence. From this, we hypothesize that $\mathbf{r}$ could be obtained by effectively manipulating two consecutive $k$-simplexes made from $\mathbf{s}$. We succeeded in extracting a polynomial of degree $k$ whose roots are $\mathbf{r}$. Thus, we turned the problem of finding $\mathbf{r}$ into the root-finding of a $k$-th order polynomial equation. After obtaining $k$ and $\mathbf{r}$, $\mathbf{a}$ can be conveniently obtained through a simple linear operation.


\subsection{Numerical Example of GSD-ST}

In this subsection, we provide a numerical example to enable the readers to understand the concept of GDS-ST. Consider the following three geometric sequences:
\begin{equation}
\begin{split}
\mathbf{s}_1 & = \{ a_1 r_1^l \}_{l=0}^{P-1} = \{ 2 \cdot 2^l \}_{l=0}^{P-1} = \{2, 4, 8, 16, 32, \cdots \}, \\
\mathbf{s}_2 & = \{ a_2 r_2^l \}_{l=0}^{P-1} = \{ 1 \cdot 3^l \}_{l=0}^{P-1} = \{1, 3, 9, 27, 81,  \cdots \}, \\
\mathbf{s}_3 & = \{ a_3 r_3^l \}_{l=0}^{P-1} = \{ 4 \cdot (-1)^l \}_{l=0}^{P-1} = \{4, -4, 4, -4, 4, \cdots \}.
\end{split}
\end{equation}
Then, suppose that we have no information on the three sequences and that we can observe only their superposition, $\mathbf{s}$, i.e.,
\begin{equation}
    \mathbf{s} = \mathbf{s}_1 + \mathbf{s}_2 + \mathbf{s}_3 = \{  2 \cdot 2^l + 1 \cdot 3^l + 4 \cdot (-1)^l \}_{l=0}^{P-1} = \{7, 3, 21, 39, 117, 303, 861, 2439, 7077 \cdots \}.
\end{equation}

Our objectives are to obtain the number of superposed sequences and the parameters of each sequence.

\subsubsection{Obtaining $k$}\label{Obtainingk}
Consider an arbitrary $\hat{k}$ as an estimate of $k$. \textbf{Theorem 1} in Section \ref{search_space} states that the volumes of successively generated $\hat{k}$-simplexes constitute a non-zero geometric sequence if and only if $\hat{k}=k$.

For the case where  $\hat{k}=2$, we consider a two-dimensional space in which we generate 2-simplexes, i.e., triangles, from the origin and consecutive values of $\mathbf{s}$. Let us create three triangles, $A_1$, $A_2$, and $A_3$, with the following coordinates:
\begin{equation}
\begin{split}
   A_1 : & [ (0,0)^{\text{T}}, (7,3)^{\text{T}}, (3,21)^{\text{T}} ], \\
   A_2 : & [ (0,0)^{\text{T}}, (3,21)^{\text{T}}, (21,39)^{\text{T}} ], \\
   A_3 : & [ (0,0)^{\text{T}}, (21,39)^{\text{T}}, (39,117)^{\text{T}} ].
\end{split}
\end{equation}
Then, we examine whether the volumes of the triangles, $\Lambda(A_n)$, constitute a geometric sequence. Because $\Lambda(A_1) = 69$, $\Lambda(A_2) = -162$, and $\Lambda(A_3) = 468$, these do not constitute a geometric sequence. Therefore, we conclude that $k \neq 2$.

For the case where $\hat{k}=3$, we increase the dimension by one and consider 3-simplexes, i.e., tetrahedrons. We create three tetrahedrons (again denoted by $A_1$, $A_2$, and $A_3$) with the following coordinates:
\begin{equation}
\begin{split}
   A_1 : & [(0,0,0)^{\text{T}}, (7,3,21)^{\text{T}}, (3, 21, 39)^{\text{T}}, (21,39, 117)^{\text{T}}], \\
   A_2 : & [(0,0,0)^{\text{T}}, (3, 21, 39)^{\text{T}}, (21,39, 117)^{\text{T}},(39, 117, 303)^{\text{T}}], \\
   A_3 : & [(0,0,0)^{\text{T}}, (21,39, 117)^{\text{T}}, (39, 117, 303)^{\text{T}}, (117,303, 861)^{\text{T}}]. \\
\end{split}
\end{equation}
Here, $\Lambda(A_1) = 192$, $\Lambda(A_2) = -1152$, and $\Lambda(A_3) = 6912$. These constitute a geometric sequence with a common ratio of $-6$. Therefore, we verify that $\mathbf{s}$ is a superposition of \textit{three} geometric sequences $(k=3)$.

For $\hat{k} > 3$, one can verify that the volumes of the $\hat{k}$-simplexes always constitute a sequence of zeros.

\subsubsection{Obtaining $\mathbf{a}$ and $\mathbf{r}$}

Given that $k$ is obtained correctly, we can fully extract the original sequences with $2k$ sampling of $\mathbf{s}$. The procedure is divided into five steps.

First, we pick $2k$ consecutive elements from $\mathbf{s}$. Second, considering a $k$-dimensional space, place $k+1$ vertices whose coordinates are $k$ basic elements of $\mathbf{s}$ (see \textbf{Definition 2} in Section \ref{basic_simplex}). In this example, four vertices are created with the coordinates
\begin{equation}
    [(7,3,21)^{\text{T}}, (3, 21, 39)^{\text{T}}, (21,39, 117)^{\text{T}}, (39,117,303)^{\text{T}}].
\end{equation}
Third, select $k$ vertices out of the $k+1$ described above. By including the origin, we can create $k+1$ $k$-simplexes in a lexicographically ordered manner. This corresponds to four tetrahedrons in this example, with the following coordinates:
\begin{equation}
\begin{split}
   B_1 : & [(0,0,0)^{\text{T}}, (7,3,21)^{\text{T}}, (3, 21, 39)^{\text{T}}, (21,39, 117)^{\text{T}}], \\
   B_2 : & [(0,0,0)^{\text{T}}, (7, 3, 21)^{\text{T}}, (3, 21, 39)^{\text{T}}, (39, 117, 303)^{\text{T}}], \\
   B_3 : & [(0,0,0)^{\text{T}}, (7, 3, 21)^{\text{T}}, (21, 39, 117)^{\text{T}}, (39, 117, 303)^{\text{T}}], \\
   B_4 : & [(0,0,0)^{\text{T}}, (3, 21, 39)^{\text{T}}, (21, 39, 117)^{\text{T}}, (39, 117, 303)^{\text{T}} ].
\end{split}
\end{equation}
Fourth, let $\Lambda(B_n)$ denote the volume of the $n$-th tetrahedron.
Surprisingly, the following relationship holds by \textbf{Theorem 2} in Section \ref{combinatorial_simplex}:
\begin{equation}\label{vol_eq}
    \left\{ \frac{\Lambda(B_1)}{\Lambda(B_1)}, \frac{\Lambda(B_2)}{\Lambda(B_1)}, \frac{\Lambda(B_3)}{\Lambda(B_1)}, \frac{\Lambda(B_4)}{\Lambda(B_1)} \right\} = \left\{ 1 , \sum_{n=1}^3 r_n, \sum_{1 \leq n < m \leq 3} r_{n}r_{m}, \prod_{n=1}^{3} r_n \right\}.
\end{equation}
Observe that these are the coefficients of a polynomial whose roots are $r_1$, $r_2$, and $r_3$. Therefore, the common ratios of the geometric sequences can be obtained by solving the polynomial equation shown below:
\begin{equation}
    x^3 - \frac{\Lambda(B_2)}{\Lambda(B_1)} x^2 + \frac{\Lambda(B_3)}{\Lambda(B_1)} x - \frac{\Lambda(B_4)}{\Lambda(B_1)} =0.
\end{equation}
Finally, once the common ratios of the sequences are obtained, we can extract the initial terms by solving a simple linear system of equations.

\subsubsection{A Case of Non-Consecutive Samples}
We have illustrated a simple example of GSD-ST with consecutive samples of $\mathbf{s}$. However, it is important to emphasize that the formal methodology of GSD-ST is more general in that the required samples need not be consecutive.

Let us select non-consecutive elements of $\mathbf{s}$ to construct the simplexes as follows:
\begin{equation}
\begin{split}
   A_1 : & [(0,0,0)^{\text{T}}, (7,3,2439)^{\text{T}}, (3, 21, 7077)^{\text{T}}, (21,39, 20703)^{\text{T}}], \\
   A_2 : & [(0,0,0)^{\text{T}}, (3, 21, 7077)^{\text{T}}, (21,39, 20703)^{\text{T}},(39, 117, 61101)^{\text{T}}], \\
   A_3 : & [(0,0,0)^{\text{T}}, (21,39, 20703)^{\text{T}}, (39, 117, 61101)^{\text{T}}, (117,303, 181239)^{\text{T}}].\\
\end{split}
\end{equation}
Notice that $\mathbf{s}[6]$, i.e., 861, is not sampled, which implies that non-consecutive samples of $\mathbf{s}$ are selected. Nevertheless, by verifying that $\{\Lambda(A_i)\}_{i=0}^2$ is a geometric sequence with a common ratio of $-6$, i.e., $\{96768, -580608, 3483648\}$, we can verify that $k=3$. Next, four vertices are created as follows:
\begin{equation}
    [(7,3,2439)^{\text{T}}, (3, 21, 7077)^{\text{T}}, (21,39, 20703)^{\text{T}}, (39,117,61101)^{\text{T}}].
\end{equation}
Then, similarly in the above example, four tetrahedrons are extracted as follows:
\begin{equation}
\begin{split}
   B_1 : & [(0,0,0)^{\text{T}}, (7,3,2439)^{\text{T}}, (3, 21, 7077)^{\text{T}}, (21,39, 20703)^{\text{T}}], \\
   B_2 : & [(0,0,0)^{\text{T}}, (7, 3, 2439)^{\text{T}}, (3, 21, 7077)^{\text{T}}, (39, 117, 61101)^{\text{T}}], \\
   B_3 : & [(0,0,0)^{\text{T}}, (7, 3, 2439)^{\text{T}}, (21, 39, 20703)^{\text{T}}, (39, 117, 61101)^{\text{T}}], \\
   B_4 : & [(0,0,0)^{\text{T}}, (3, 21, 7077)^{\text{T}}, (21, 39, 20703)^{\text{T}}, (39, 117, 61101)^{\text{T}} ].
\end{split}
\end{equation}
These new tetrahedrons also satisfy \eqref{vol_eq}, although $\mathbf{s}[5]$ and $\mathbf{s}[6]$ (i.e., 303 and 861) are not used. There are many ways of selecting the elements of $\mathbf{s}$ in a non-consecutive manner. However, it is not random and needs to comply with \textbf{Condition 2} described in Section \ref{combinatorial_simplex}.

For simplicity, the example in this section used only real numbers. However, GSD-ST is effective for complex numbers as well. Another example of complex-valued geometric sequences is presented in Appendix A. Furthermore, source code for the examples is presented in \cite{leeCode}.

\section{Methodology of GSD-ST}

\begin{figure*}[!]
\centering
\includegraphics[width=0.97\textwidth]{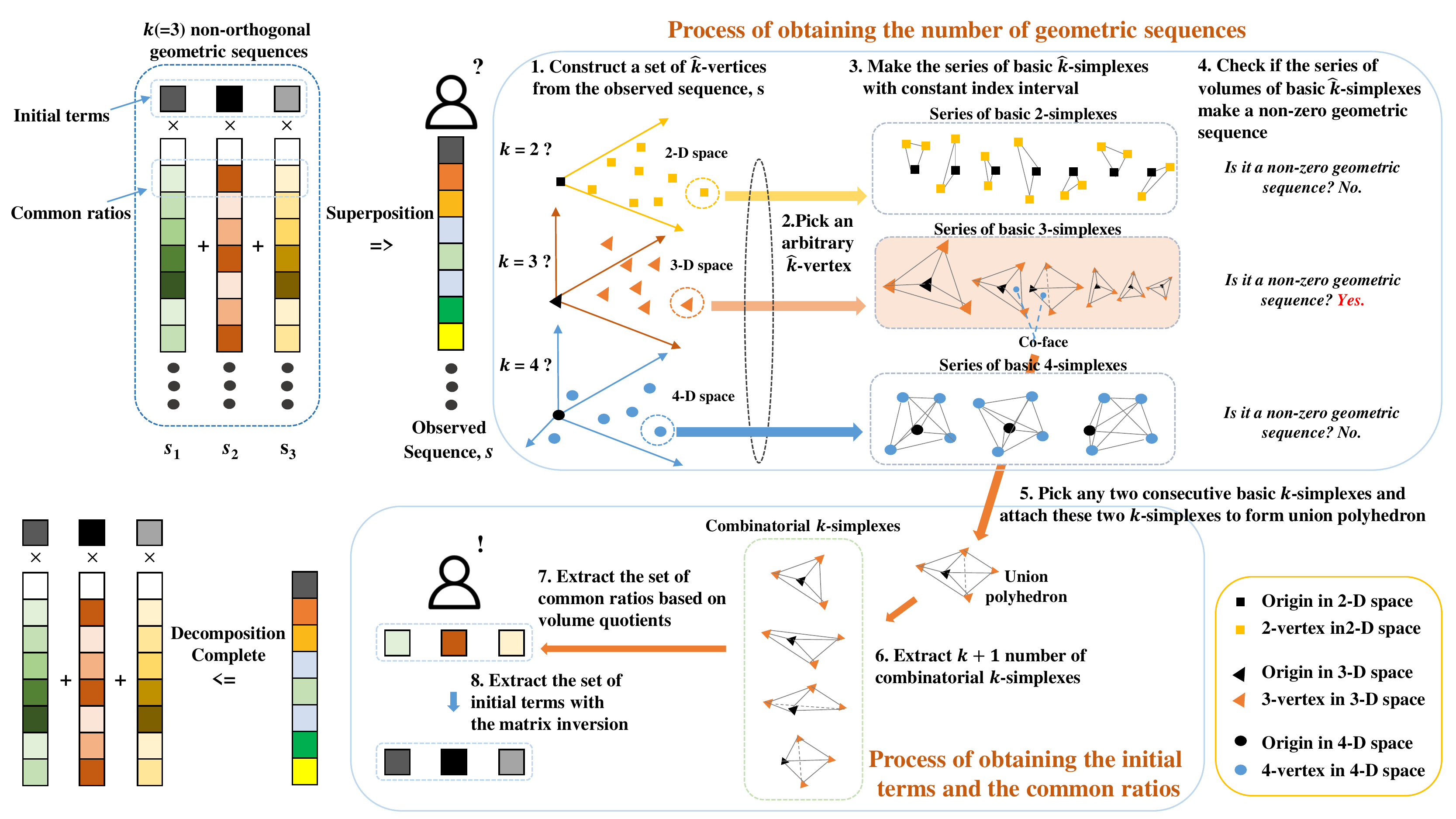}
\caption{Illustration of concept of GSD-ST. It is composed mainly of three steps: i) transformation of $\mathbf{s}$ to the set of $\hat{k}$-vertices, ii) derivation of the series of volumes of the basic $k$-simplexes, and iii) extraction of the volume quotients of the combinatorial $k$-simplexes.}
\label{GSDST}
\end{figure*}

This section provides the general methodology of GSD-ST. The overall concept of GSD-ST is depicted in Fig. \ref{GSDST}.

\subsection{Set of $\hat{k}$-Vertices to Construct a Search Space}
\label{search_space}
At this stage, $\mathbf{s}$ is the only observable sequence. Because $k$ is unknown, let us assume an arbitrary $\hat{k}$. We construct a search space as follows to obtain the unknowns, i.e., $k$, $\mathbf{a}$, and $\mathbf{r}$:

\begin{defn}
For arbitrary $i_c \in \mathbb{N}$ and $\phi_{\hat{k}}$, let a $\hat{k}$-dimensional search space, $\Xi_{i_c}(\phi_{\hat{k},\mathbf{s}})$, be the lexicographically-ordered collection of the vertices that are formed by the successive spawning of new $\hat{k}$-vertices with the index-shifting of $i_c \cdot \mathbf{1}_{\hat{k}}$ starting from $\phi_{\hat{k},\mathbf{s}}$.
\end{defn}

For example, if $\hat{k}=3$, $i_c = 2$, and $\phi_{\hat{k}}=\{0,1,4\}$,
\begin{equation}
    \Xi_{2}(\phi_{{3},(\mathbf{s}[0],\mathbf{s}[1],\mathbf{s}[4])^{\text{T}}}) = \{ (\mathbf{s}[0],\mathbf{s}[1],\mathbf{s}[4])^{\text{T}}, (\mathbf{s}[2],\mathbf{s}[3],\mathbf{s}[6])^{\text{T}}, \cdots \}.
\end{equation}
By the definition of a geometric sequence, $\mathbf{s}[n]$ is a polynomial of degree $n+1$ consisting of the initial terms of degree one and common ratios of degree $n$. Consider the volume of an arbitrary $\hat{k}$-simplex formed by the origin and $\hat{k}$ consecutive vertices in $\Xi_{i_c} (\phi_{\hat{k},\mathbf{s}})$. That is, $\Lambda (e(\Xi_{i_c} (\phi_{\hat{k},\mathbf{s}})[j], \Xi_{i_c} (\phi_{\hat{k},\mathbf{s}})[j+1], \cdots , \Xi_{i_c} (\phi_{\hat{k},\mathbf{s}})[j+\hat{k}-1]))$ for $j \in \mathbb{N}_0$, where $\Xi_{i_c} (\phi_{\hat{k},\mathbf{s}})[j]$ is the $j$-th vertex in $\Xi_{i_c} (\phi_{\hat{k},\mathbf{s}})$. According to Definition 1, it becomes a homogeneous polynomial whose degree is determined by $j$, $\hat{k}$, $i_c$, and $\phi_{\hat{k}}$. This property of algebraic geometry supports the approaches in the following subsections.

\subsection{Series of Basic $\hat{k}$-simplexes to Obtain the Number of Superposed Geometric Sequences}
\label{basic_simplex}
In this subsection, we obtain $k$. To achieve this, let us define a basic $\hat{k}$-simplex as follows:

\begin{defn}
For an arbitrary $j\in \mathbb{N}_0$, a basic $\hat{k}$-simplex is defined by \begin{equation}
e(\Xi_{i_c} (\phi_{\hat{k},\mathbf{s}})[j], \cdots , \Xi_{i_c} (\phi_{\hat{k},\mathbf{s}})[j+\hat{k}-1]) \in \mathbb{C}^{\hat{k} \times (\hat{k}+1)}.
\end{equation}
Furthermore, let ${\xi}_{i_c}(\phi_{\hat{k},\mathbf{s}})$ be a series of the basic $\hat{k}$-simplexes which is formed by the lexicographically ordered collection of the basic $\hat{k}$-simplexes over $j=0,1,\cdots$.
\end{defn}

In addition, recall that $\Lambda({\xi}_{i_c}(\phi_{\hat{k},\mathbf{s}}))$ denotes a series of volumes of basic $\hat{k}$-simplexes. With Definition 2, the search space $\Xi_{i_c}(\phi_{\hat{k},\mathbf{s}})$ exhibits a noteworthy property if we estimate $k$ correctly, i.e., $\hat{k}=k$. This is specified in the following Lemma:

\begin{lem}
Given $\mathbf{s}$, for arbitrary $i_c \in \mathbb{N}$ and $\phi_{\hat{k}}$, $\Lambda ({\xi}_{i_c}(\phi_{\hat{k},\mathbf{s}}))$ is a non-zero geometric sequence whose common ratio is $\Big( \prod_{{n}=1}^{\hat{k}} r_n \Big) ^{i_c}$ if $\hat{k}=k$.
\end{lem}
\begin{proof}
To establish Lemma 1, we utilize the fact that the ratio between any two consecutive samples of a geometric sequence is a constant. Hence, we demonstrate the following for an arbitrary $j\in \mathbb{N}_0$:
\begin{equation}
    \frac{\Lambda(\xi_{i_c}(\phi_{k,\mathbf{s}}))[j+1]}{\Lambda(\xi_{i_c}(\phi_{k,\mathbf{s}}))[j]} = C,
\end{equation}
where $C\in\mathbb{C}$ is a non-zero constant.

To derive $\Lambda(\xi_{i_c}(\phi_{k,\mathbf{s}}))[j]$, let us define $\mathbf{\Omega}_j \in \mathbb{C}^{k \times k}$ as the matrix form of $\xi_{i_c}(\phi_{k,\mathbf{s}})[j]$ excluding the origin. That is, $\mathbf{\Omega}_j$ is the square matrix whose $l$-th column is the coordinate of the $l$-th vertex in $\xi_{i_c}(\phi_{k,\mathbf{s}})[j]$.
Then, \begin{equation}
    {\Lambda(\xi_{i_c}(\phi_{k,\mathbf{s}}))[j] = \frac{1}{k!}\det(\mathbf{\Omega}_j).}
\end{equation}
Recalling the definition of $\phi_k$ and $i_c$, $\mathbf{\Omega}_j$ can be represented as follows:
\begin{gather}\label{vol_j}
    \mathbf{\Omega}_j =
\begin{bmatrix}
    \mathbf{s}[i_cj + \phi_k[0]] & \mathbf{s}[i_cj + \phi_k[0] + i_c] & \cdots & \mathbf{s}[i_cj + \phi_k[0] + (k-1)i_c] \\
    \mathbf{s}[i_cj + \phi_k[1]] & \mathbf{s}[i_cj + \phi_k[1] + i_c] & \cdots & \mathbf{s}[i_cj + \phi_k[1] + (k-1)i_c] \\
    \vdots & \vdots & \ddots & \vdots \\
    \mathbf{s}[i_cj + \phi_k[k-1]] & \mathbf{s}[i_cj + \phi_k[k-1] + i_c] & \cdots & \mathbf{s}[i_cj + \phi_k[k-1] + (k-1)i_c]
\end{bmatrix}
\end{gather}
\\
Here, we can decompose $\mathbf{\Omega}_j$ as follows with regard to the degree of each element:
\begin{gather}\label{vol_j}
    \mathbf{\Omega}_j = \begin{bmatrix}
    \sum_{i=1}^k a_i r_i^{i_cj + \phi_k[0]} & \sum_{i=1}^k a_i r_i^{i_cj + \phi_k[0] + i_c} & \cdots & \sum_{i=1}^k a_i r_i^{i_cj + \phi_k[0] + (k-1)i_c} \\
    \sum_{i=1}^k a_i r_i^{i_cj + \phi_k[1]} & \sum_{i=1}^k a_i r_i^{i_cj + \phi_k[1] + i_c} & \cdots & \sum_{i=1}^k a_i r_i^{i_cj + \phi_k[1] + (k-1)i_c} \\
    \vdots & \vdots & \ddots & \vdots \\
    \sum_{i=1}^k a_i r_i^{i_cj + \phi_k[k-1]} & \sum_{i=1}^k a_i r_i^{i_cj + \phi_k[k-1] + i_c} & \cdots & \sum_{i=1}^k a_i r_i^{i_cj + \phi_k[k-1] + (k-1)i_c}
\end{bmatrix} \\
    =
\begin{bmatrix}
    r_1^{\phi_k [0]} & \cdots & r_k^{\phi_k [0]}  \\
    r_1^{{\phi_k [1]}} & \cdots & r_k^{\phi_k [1]} \\
    \vdots & \ddots & \vdots\\
    r_1^{{\phi_k [k-1]}} & \cdots &  r_k^{{\phi_k [k-1]}}
\end{bmatrix}
\cdot
\begin{bmatrix}
    a_1 &  \cdots & 0\\
    \vdots & \ddots & \vdots\\
    0 & \cdots &  a_k
\end{bmatrix}
\cdot
\begin{bmatrix}
    r_1 &  \cdots & 0\\
    \vdots & \ddots & \vdots\\
    0 & \cdots & r_k
\end{bmatrix} ^{i_cj}
\cdot
\begin{bmatrix} \label{finalD}
    r_1^{0} & \cdots & r_k^0  \\
    r_1^{i_c} & \cdots & r_k^{i_c} \\
    \vdots & \ddots & \vdots\\
    r_1^{(k-1)i_c} & \cdots &  r_k^{(k-1)i_c}
\end{bmatrix}^{\text{T}}.
\end{gather}
For brevity, let $\mathbf{\Phi}$, $\mathbf{\Sigma}_{\mathbf{a}}$, $\mathbf{\Sigma}_{\mathbf{r}}$, and $\mathbf{\Psi}$ denote the above four factors, respectively, as follows: \\ \begin{gather}\label{finalD2}
\mathbf{\Phi} =
  \begin{bmatrix}
    r_1^{\phi_k [0]} & \cdots & r_k^{\phi_k [0]}  \\
    r_1^{{\phi_k [1]}} & \cdots & r_k^{\phi_k [1]} \\
    \vdots & \ddots & \vdots\\
    r_1^{{\phi_k [k-1]}} & \cdots &  r_k^{{\phi_k [k-1]}}
\end{bmatrix},
  \mathbf{\Sigma}_{\mathbf{a}} = \begin{bmatrix}
    a_1 &  \cdots & 0\\
    \vdots & \ddots & \vdots\\
    0 & \cdots &  a_k
\end{bmatrix},
  \mathbf{\Sigma}_{\mathbf{r}} = \begin{bmatrix}
    r_1 &  \cdots & 0\\
    \vdots & \ddots & \vdots\\
    0 & \cdots & r_k
\end{bmatrix}
,
  \mathbf{\Psi} = \begin{bmatrix}
    r_1^{0} & \cdots & r_k^0  \\
    r_1^{i_c} & \cdots & r_k^{i_c} \\
    \vdots & \ddots & \vdots\\
    r_1^{(k-1)i_c} & \cdots &  r_k^{(k-1)i_c}
\end{bmatrix}.
\end{gather}
Therefore, the decomposition in \eqref{finalD} can be represented by
$\mathbf{\Omega}_j =
    \mathbf{\Phi}
    \mathbf{\Sigma}_{\mathbf{a}}
    \mathbf{\Sigma}_{\mathbf{r}}^{i_cj}
    \mathbf{\Psi}^{\text{T}}$.
Observe that only $\mathbf{\Sigma}_{\mathbf{r}}^{i_c j}$ depends on $j$. Hence, the following equality is satisfied for any $j$:
\begin{equation}\label{detRatio}
\frac{\Lambda ({\xi}_{i_c}(\phi_{k,\mathbf{s}}))[j+1]}{\Lambda ({\xi}_{i_c}(\phi_{k,\mathbf{s}}))[j]} = \frac{\det(\mathbf{\Omega}_{j+1})/{k!}}{\det(\mathbf{\Omega}_{j})/{k!}} = \frac{\det((\mathbf{\Sigma}_{\mathbf{r}})^{i_c(j+1)})}{\det((\mathbf{\Sigma}_{\mathbf{r}})^{i_c j})} = \Big(\det(\mathbf{\Sigma_{r}})\Big)^{i_c} =\Big( \prod_{n =1}^k r_n \Big) ^{i_c}.
\end{equation}
Therefore, we can verify that $\Lambda(\xi_{i_c}(\phi_{k,\mathbf{s}}))$ is a geometric sequence.
\end{proof}

It is a noteworthy and effective property that $k$ non-orthogonally superposed geometric sequences can be transformed into a geometric sequence. The fact that $\Lambda ({\xi}_{i_c}(\phi_{k,\mathbf{s}}))$ is a geometric sequence, regardless of $i_c$ and $\phi_{k}$, can be utilized for acquiring $k$.

\begin{thm}
Given $\mathbf{s}$, for arbitrary $i_c \in \mathbb{N}$ and $\phi_{\hat{k}}$, $\Lambda({\xi}_{i_c}(\phi_{\hat{k},\mathbf{s}}))$ is a non-zero geometric sequence if and only if $\hat{k}=k$.
\end{thm}
\begin{proof}
Lemma 1 provides that $\Lambda({\xi}_{i_c}(\phi_{\hat{k},\mathbf{s}}))$ is a non-zero geometric sequence if $\hat{k}=k$.
Additionally, we present its inverse as follows:
\begin{enumerate}[label=\roman*)]
    \item $\hat{k}<k$:
    recalling \eqref{finalD} and \eqref{finalD2}, $\mathbf{\Omega}_j$ can also be decomposed as $\mathbf{\Phi}
    \mathbf{\Sigma}_{\mathbf{a}}
    \mathbf{\Sigma}_{\mathbf{r}}^{i_cj}
    \mathbf{\Psi}^{\text{T}}$ even if $\hat{k}<k$. However, in this case, $\mathbf{\Phi} \in \mathbb{C}^{\hat{k}\times k}$ and $\mathbf{\Psi}\in \mathbb{C}^{\hat{k}\times k}$ lead to an underdetermined system. Thus, $\det( \mathbf{\Sigma}_{\mathbf{r}}^{i_cj})$ cannot be the factor of $\det(\mathbf{\Omega}_j)$ owing to the rank deficiency. Therefore, $\frac{\det(\mathbf{\Omega}_{j+1})}{\det(\mathbf{\Omega}_j)}$ is still a function of $j$. Hence, $\Lambda({\xi}_{i_c}(\phi_{\hat{k},\mathbf{s}}))$ is not a geometric sequence.
    \item $\hat{k}>k$: $\hat{k}$-simplexes cannot be represented by $k$ linearly independent bases. Therefore, the corresponding volumes become zero, i.e., $  \Lambda({\xi}_{i_c}(\phi_{\hat{k},\mathbf{s}})) $ is an all-zero sequence.
\end{enumerate}
\end{proof}

To be able to examine whether $\hat{k}=k$, the following condition should be fulfilled:
\begin{con}\label{con1}
Given $\mathbf{s}$, $\text{card}(\Xi_{i_c}(\phi_{\hat{k},\mathbf{s}}))$ must be larger than $k+1$ for obtaining $k$ regardless of $i_c$ and $\phi_{\hat{k}}$.
\end{con}

As a special case of Condition 1, we derive the minimum number of samples to obtain $k$, i.e., the minimum required $P$.

\begin{cor}
Given $\mathbf{s}$, the minimum required $P$ to obtain $k$ is $2k+1$.
\end{cor}
\begin{proof}
The required $P$ is minimized when we minimize the number of unsampled elements in $s$. This implies that $i_c=1$ and $\phi_k=\{0,1,\cdots, k-1 \}$. By applying Condition 1 to this setting, the minimum $\Xi_{i_c}(\phi_{k,\mathbf{s}})$ is composed of $\{\mathbf{s}[0], \mathbf{s}[1], \cdots , \mathbf{s}[2k]\}$. This concludes the proof.
\end{proof}

\subsection{Combinatorial $k$-simplexes to Extract Initial Terms and Common Ratios}
\label{combinatorial_simplex}

After obtaining $k$ in $\mathbf{s}$, we can specify the search space for extracting $\mathbf{a}$ and $\mathbf{r}$.

\begin{defn}
For an arbitrary $\phi_{k}$, let the $j$-th series of combinatorial $k$-simplexes, ${\varkappa}_j(\phi_{k,\mathbf{s}})$, be the output of the following process:
\begin{enumerate} [label=\roman*)]
    \item Fix the search space to $\Xi_1 (\phi_{k,\mathbf{s}})$ and construct ${\xi}_1 (\phi_{k,\mathbf{s}})$ over ${\Xi}_1 (\phi_{k,\mathbf{s}})$.
    \item Pick any two consecutive basic $k$-simplexes such as $\xi_1 (\phi_{k,\mathbf{s}}){[j]}$ and $\xi_1 (\phi_{k,\mathbf{s}}){[j+1]}$ from ${\xi}_1 (\phi_{k,\mathbf{s}})$.
    \item Paste these two $k$-simplexes to create a new polyhedron having $k+2$ vertices, which we call the $j$-th \emph{union polyhedron} \footnote{These two consecutive basic $k$-simplexes can be attached because they share the coface made by $k$ $k$-vertices based on the definition of ${\xi}_{1}(\phi_{k,\mathbf{s}})$.}.
    \item Extract lexicographically ordered $k+1$ $k$-simplexes out of the $j$-th union polyhedron.
\end{enumerate}
\end{defn}

We further define the volume quotients of the combinatorial $k$-simplexes with the $j$-th union polyhedron, $\mathbf{v}_j(\phi_{k,\mathbf{s}})$, as follows:
\begin{equation}\label{vj}
 \mathbf{v}_j(\phi_{k,\mathbf{s}}):= \{ \frac{\Lambda(\varkappa_j(\phi_{k,\mathbf{s}})[0])}{\Lambda(\varkappa_j(\phi_{k,\mathbf{s}})[0])}, \cdots ,
\frac{\Lambda(\varkappa_j(\phi_{k,\mathbf{s}})[k])}{\Lambda(\varkappa_j(\phi_{k,\mathbf{s}})[0])} \}.
\end{equation}

\begin{thm}
{Given $\mathbf{s}$, regardless of $j$ and $\phi_{k}$, $\mathbf{v}_j(\phi_{k,\mathbf{s}})$ is unique as follows:
\begin{equation}
\{ 1 , \cdots , \sum_{1 \leq i_1 < i_2 < \dots < i_l \leq k}   \big( \prod_{n=1}^l r_{i_n} \big), \cdots , \prod_{n=1}^{k} r_n \}.
\end{equation}}
\end{thm}
\begin{proof}
{See Appendix B.}
\end{proof}

By Theorem 2, $\mathbf{v}_j(\phi_{k,\mathbf{s}})$ is unique for a given $\mathbf{s}$, and therefore, can be simplified as $\mathbf{v}(k, \mathbf{s})$. This uniqueness is a strong property because it implies that all the search spaces contain identical information regardless of the selection of the initial $k$-vertex.

Based on the information of $\mathbf{v}(k, \mathbf{s})$, we can construct a polynomial equation for $\mathbf{r}$ as follows:
\begin{equation}\label{poly_eq}
    \sum_{n=0} ^k \big(    (-1)^{k-n} \cdot \mathbf{v}(k,\mathbf{s})[k-n]\cdot r^{n} \big) = 0.
\end{equation}
The roots of \eqref{poly_eq} are the common ratios $\mathbf{r}$, which we seek. Furthermore, the minimum number of samples for extracting $\mathbf{r}$ can be described as follows:

\begin{cor}
Given $k$, the most compact sampling for extracting $\mathbf{r}$ is to take $2k$ consecutive samples of the sequence $\mathbf{s}$.
\end{cor}
\begin{proof}
This is equivalent to constructing $\mathbf{v}(k, \mathbf{s})$ with the least number of $k$-vertices made by $\phi_k$ of $\{0,1,\cdots, k-1 \}$. Then, the minimum required  $\Xi_{1}(\phi_{k,\mathbf{s}})$ is composed of $\{\mathbf{s}[0], \mathbf{s}[1], \cdots , \mathbf{s}[2k-1]\}$. This concludes the proof.
\end{proof}

After obtaining $\mathbf{r}$, it is trivial to extract $\mathbf{a}$ by simple matrix pseudo-inversion: $\mathbf{a}=\mathbf{R}^+ \mathbf{s}$. Here, $\mathbf{R} \in \mathbb{C}^{P \times k}$ is the matrix constructed by $\mathbf{r}$ and satisfies $\mathbf{R}[m,n]:=r_{n+1}^{m}$ for $m,n \in \mathbf{N}_0$, and $(\cdot)^+$ is the pseudo-inverse operation.
Each pair of initial term and common ratio is matched through this matrix operation. Thus, there is no paring problem between the initial terms and common ratios.

Note that the following condition should be satisfied to obtain $\mathbf{a}$ and $\mathbf{r}$ with the knowledge of $k$.
\begin{con} Given $\mathbf{s}$,
$\text{card}(\Xi_1(\phi_{k,\mathbf{s}}))$ must be larger than $k$ to obtain $\mathbf{a}$ and $\mathbf{r}$ regardless of $\phi_{k}$. \end{con}

The above condition provides a noteworthy characteristic in terms of the non-consecutive and non-uniform sampling.
For example, when $k=3$, the initial terms and common ratios can be extracted by strangely sampled observations such as $\{ \mathbf{s}[0]$, $\mathbf{s}[1]$, $\mathbf{s}[2]$, $\mathbf{s}[3]$, $\mathbf{s}[20]$, $\mathbf{s}[21]$, $\mathbf{s}[22]$, $\mathbf{s}[23]$, $\mathbf{s}[100]$, $\mathbf{s}[101]$, $\mathbf{s}[102]$, $\mathbf{s}[103]\}$, i.e., $\phi_{3,\mathbf{s}}=(\mathbf{s}[0],\mathbf{s}[20],\mathbf{s}[100])^{\text{T}}$. Furthermore, the GSD-ST method can be extended to an arbitrary $k$-polytope. Because all the $k$-simplexes in this work include the origin point, the volume of any $k$-polytope can be represented by addition and/or subtraction operations of the volumes of these $k$-simplexes.

To summarize, starting from an arbitrary $\hat{k}$-vertex, $\phi_{\hat{k},\mathbf{s}}$, we identify $k$ by verifying that a series of volumes of basic ${k}$-simplexes, $\Lambda ({\xi}_{i_c} (\phi_{{k},\mathbf{s}}))$, is a geometric sequence. Then, we obtain $\mathbf{r}$ through the volume quotients of the combinatorial $k$-simplexes, $\mathbf{v}(k,\mathbf{s})$. Finally, $\mathbf{a}$ is obtained by a simple matrix operation. Algorithm 1 depicts the GSD-ST procedure, which consists of two phases: acquisition of $k$ and extraction of $\mathbf{a}$ and $\mathbf{r}$.

\begin{algorithm}[t!]
 \caption{: GSD-ST process, $\mathcal{S}(\cdot)$}\label{alg:rmleft}
 \begin{algorithmic}[1]
\State Set the observed sequence, $\mathbf{s}:=\sum_{i=1}^k \mathbf{s}_i$.
\State [\textit{Phase 1}: {Acquisition of $k$}]
\State Set $\hat{k}$ to 1.
\While {$\Lambda({\xi}_{i_c}(\phi_{\hat{k},\mathbf{s}}))$ is not a geometric sequence}
\State Set $\hat{k}$ to $\hat{k}+1$.
\State Pick arbitrary ${i_c} \in \mathbb{N}$ and $\phi_{\hat{k}}$.
\State Establish $\Xi_{i_c}(\phi_{\hat{k},\mathbf{s}})$ as per Definition 1.
\State Establish ${\xi}_{i_c}(\phi_{\hat{k},\mathbf{s}})$ as per Definition 2.
\State Construct $\Lambda({\xi}_{i_c}(\phi_{\hat{k},\mathbf{s}}))$.
\EndWhile
\State Set $k$ to $\hat{k}$
\State [\textit{Phase 2}: Extraction of $\mathbf{a}$ and $\mathbf{r}$]
\State Construct $\mathbf{v}(k,\mathbf{s})$ from $\Xi_1(\phi_{k,\mathbf{s}})$.
\State Extract $\mathbf{r}$ by finding the roots of \eqref{poly_eq}.
\State Extract $\mathbf{a}$ by $\mathbf{a}=\mathbf{R}^+ \mathbf{s}$ where $\mathbf{R}[m,n]:=r_{n+1}^m$.
 \end{algorithmic}
\end{algorithm}

Let the GSD-ST method be denoted by $\mathcal{S}(\cdot)$. Then,
\begin{equation}
    \mathcal{S}(\mathbf{s}):=\{ (a_1, r_1), \cdots, (a_k, r_k) \},
\end{equation}
where $\mathcal{S}$ is a nonlinear function including the entire process of obtaining $k$ and extracting $\{ (a_1, r_1), \cdots, (a_k, r_k) \}$. In addition, we define the inverse of GSD-ST, $\mathcal{S}^{-1}(\cdot)$, as follows:
\begin{equation}
    \mathcal{S}^{-1}(\{ (a_1, r_1), \cdots, (a_k, r_k) \}):= \mathbf{s}.
\end{equation}
If Conditions 1 and 2 hold, $\mathcal{S}^{-1}(\mathcal{S}(\mathbf{s}))$ is equivalent to $\mathbf{s}$.

\subsection{GSD-ST with Noisy Samples}
Until now, we have assumed that the observation of $\mathbf{s}$ is flawless. However, the observed sequence may be prone to noise, particularly in wireless communications. Therefore, in this section, we present practical methods for mitigating errors in $\mathbf{s}$. Let us define $\mathbf{s}_w$ such that
\begin{equation}
    \mathbf{s}_w:= \mathbf{s} + \mathbf{w} = \{ \mathbf{s}[l] + \mathbf{w}[l] \}_{l=0}^{P-1},
\end{equation}
where $\mathbf{w}$ denotes a sequence of random variables representing additive white Gaussian noise (AWGN). The fundamental approach of the de-noising process is to utilize more samples than the minimum requirement, i.e., to consider $P > 2k+1$.

\subsubsection{Estimation of $k$ with noise}
We can design an approximated algorithm to estimate $k$ by utilizing Theorem 2, i.e., all possible $\mathbf{v}(\hat{k},\mathbf{s}_w) \in \mathbb{C}^{\hat{k}+1}$ are identical when $\hat{k}=k$ and $\mathbf{s}_w=\mathbf{s}$. To extract the informative part of $\mathbf{v}(\hat{k},\mathbf{s}_w)$, $\mathbf{v}_{I,\hat{k}} \in \mathbb{C}^{\hat{k}}$ is defined as $ \mathbf{v}(\hat{k},\mathbf{s}_w) \setminus \mathbf{v}(\hat{k},\mathbf{s}_w)[0] $ because $\mathbf{v}(\hat{k},\mathbf{s}_w)[0]=1$ for any $\mathbf{v}(\hat{k},\mathbf{s}_w)$. Here, $\setminus$ is the operator of set minus.

We use a well-known method based on Euclidean distance to examine the similarity among all possible $\mathbf{v}(\hat{k},\mathbf{s}_w)$. Let $\mathbf{v}_{I,\hat{k}}^{(i)}$ and $\mathbf{v}_{I,\hat{k}}^{(j)}$ be two $\mathbf{v}_{I,\hat{k}}$ among all the possible cases. Then, the similarity function $D(\hat{k},\mathbf{s}_w)$ is defined as follows:
\begin{equation}\label{similarity}
 D(\hat{k}, \mathbf{s}_w):= \Big( \prod_{1 \leq i_c \leq i_U, 1 \leq i  < j \leq {i_K}} ||\mathbf{v}_{I,\hat{k}}^{(i)} - \mathbf{v}_{I,\hat{k}}^{(j)}||\Big) ^ {  \frac{1}{ {  {  \sum_{i_c = 1}^{i_U}{i_K \choose 2 }}} } },
\end{equation}
where $|| \cdot ||$ is the Euclidean norm of an input. Furthermore, $i_U$ is the upper bound of $i_c$, which corresponds to $\floor{\frac{P-\hat{k}}{\hat{k}+1}}$. Here, $\floor{\cdot}$ is the operator of the floor calculation. Given $i_c$, $i_K$ denotes the number of all possible $\mathbf{v}_{I,\hat{k}}$, ${{P-i_c \hat{k}} \choose \hat{k}}$. Thus, $D(\hat{k},\mathbf{s}_w)$ is the geometric mean of the similarity values for all possible $\mathbf{v}_{I,\hat{k}}$. Let $k^*$ be $\hat{k}$ minimizing $D(\hat{k},\mathbf{s}_w)$ and we consider it to be $k$. If $k^*=k$ and $\mathbf{s}_w = \mathbf{s}$, $D(k^*,\mathbf{s}_w)$ becomes zero.

Let $N_d$ be the number of Euclidean distances to be computed. It is given by
\begin{equation}\label{N_d}
  N_d = {  {  \sum_{i_c = 1}^{i_U}{{{P-i_c \hat{k}} \choose \hat{k}} \choose 2 }}}.
\end{equation}
The computation of $D(\hat{k}, \mathbf{s}_w)$ is demanding when the number of $P$ is large. Thus, we define two simplified similarity functions, $D_d(\hat{k},\mathbf{s}_w)$ and $D_r(\hat{k},\mathbf{s}_w)$, as follows:

\begin{enumerate}[label=\roman*)]
    \item $D_{d}(\hat{k},\mathbf{s}_w):=
\Big( \prod_{1 \leq i < j \leq {{{P- \hat{k}} \choose \hat{k}}}} ||\mathbf{v}_{I,\hat{k}}^{(i)} - \mathbf{v}_{I,\hat{k}}^{(j)}||\Big) ^ {  \frac{1}{ { {  {{{P- \hat{k}} \choose \hat{k}} \choose 2 }}} } }
$,
    \item $D_{r}(\hat{k},\mathbf{s}_w):= ||\mathbf{v}_{I,\hat{k}}^{(i)} - \mathbf{v}_{I,\hat{k}}^{(j)})||$.
\end{enumerate}
$D_d(\hat{k},\mathbf{s}_w)$ is simplified by fixing $i_U$ as $1$ in \eqref{N_d}.
Furthermore, $D_r (\hat{k},\mathbf{s}_w)$ is the simplest method that executes only one calculation of Euclidean distance with arbitrary $\mathbf{v}_{I,\hat{k}}^{(i)}$ and $\mathbf{v}_{I,\hat{k}}^{(j)}$.

\subsubsection{Extraction of $\mathbf{a}$ and $\mathbf{r}$ with noise}
The de-noising of $\mathbf{s}_w$ is essentially a process of separating $\mathbf{s}$ and $\mathbf{w}$. Here, $\mathbf{s}$ is an index-wise correlated sequence with $2k$ parameters of interest, whereas $\mathbf{w}$ is a sequence of random variables. We focus on the fact that $\mathbf{s}$ has $k$ non-identical bases. Furthermore, each basis is formed by only one parameter, i.e., the common ratio, from the characteristic of the geometric sequence. From this perspective, the de-noising of $\mathbf{s}_w$ can be handled by suppressing the number of bases for $\mathbf{s}_w$ to $k$.

Thus, we utilize iterative $k$-truncated singular value decomposition (SVD) \cite{hansen1987truncatedsvd} for taking the $k$ largest singular values and their corresponding vectors.
Let ${\mathbf{s}_w}^*$ be the de-noised sequence which is the output of the following process:
\begin{enumerate}[label=\roman*)]
    \item Create a matrix, $\mathbf{Q} \in \mathbb{C}^{P_h \times (P-P_h+1)}$, where $P_h=\floor{\frac{P+1}{2}}$ from $\mathbf{s}_w$. $\mathbf{Q}$ satisfies the condition:
    \begin{equation}
        \mathbf{Q}[m,n]:=\mathbf{s}_w[m+n].
    \end{equation}
    \item Execute the $k$-truncated SVD of $\mathbf{Q}$.
    \item Reconstruct $\mathbf{Q}$ using the $k$-tuples of singular values and vectors.
    \item Transform $\mathbf{Q}$ onto a new $\mathbf{s}_w$ by averaging the values with the same index, i.e.,
    \begin{equation}
        \{ \mathbf{s}_w[l] \}_{l=0}^{P-1} := \frac{1}{q} \sum_{n,m} \mathbf{Q}[m,n], ~ \text{s.t.} ~ m+n=l,
    \end{equation}
    where $q=l+1$ if $l<P_h$, and $q=P-l$ otherwise.
    \item Repeat the above four-step process with the stopping criterion $\epsilon$ until $\mathbf{s}_w$ converges to $\mathbf{s}_w^*$.
\end{enumerate}

The above de-noising process for $\mathbf{s}_w$ is equivalent to making the bases of $\mathbf{Q}$ for both row space and column space identical to each other, with the aim of minimizing the number of parameters in $\mathbf{s}_w$.
Then, we extract $\mathbf{a}$ and $\mathbf{r}$ using $\mathbf{s}_w^*$.
Algorithm 2 summarizes the process of GSD-ST with noisy samples.

\begin{algorithm}[t!]
 \caption{: GSD-ST process with noisy samples}\label{alg:rmleft}
 \begin{algorithmic}[1]
\State [\textit{Phase 1}: {Estimation of $k$}]
\State Determine $k^*$ that minimizes $D(\hat{k},\mathbf{s}_w)$ in \eqref{similarity}.
\State [\textit{Phase 2}: Extraction of $\mathbf{a}$ and $\mathbf{r}$]
\State Obtain $\mathbf{s}_w^*$ by iterative $k$-truncated SVD of $\mathbf{Q}$.
\State Implement Phase 2 of Algorithm 1 for $\mathbf{s}_w^*$.
 \end{algorithmic}
\end{algorithm}

\section{Application of GSD-ST to Non-orthogonal Interference-free radio Access}

\subsection{ {Potential of GSD-ST for Wireless Communications}}

As we discussed in the Introduction, the equidistant samples of a radio wave comprise a geometric sequence. This may not be apparent for the case of quadrature amplitude modulation (QAM) because the radio wave is discontinuous over time. However, if we consider a symbol duration, the modulated signal can be interpreted as a continuous wave that contains the modulation information in the initial term. Hence, an accumulation of radio waves in a symbol duration is equivalent to a superposition of geometric sequences. Therefore, the capability to decompose geometric sequences offers the potential for separating (i.e., demodulating) non-orthogonally superposed radio waves.


The use of orthogonality has been the fundamental method for handling multiple radio waves. Orthogonality is available in various domains such as time, frequency, and space. Recently, orthogonality has been sought in a more sophisticated domain such as codebook \cite{nikopour2013sparse}. In general, there is no guarantee that the proposed GSD-ST would yield a higher performance than the existing multiple access schemes. This is particularly so when it is compared with a well-designed scheme, e.g., sparse code multiple access (SCMA) \cite{nikopour2013sparse}. However, orthogonality is not always feasible in wireless communications. For example, consider random access in cellular systems, where the transmissions of multiple users cannot be coordinated. Typical random access schemes arrange a finite number of orthogonal resources from which each user selects randomly. This inevitably incurs interference owing to collisions regardless of the number of orthogonal resources.

We can pursue a different approach to random access, i.e., non-orthogonal transmissions with GSD-ST. Assume that each transmitter randomly selects its frequency in a specified bandwidth. Here, a fundamental difference from the existing schemes is that the frequency is selected in a continuous domain rather than from a finite grid. Even if the overlapping radio waves are not orthogonal, they can be decomposed as if there were no interference by using only $2k+1$ sampling at a rate faster than the highest frequency component. Theoretically, GSD-ST enables the infinitely many users to share a limited bandwidth because the probability of randomly selecting an identical continuous number is \emph{zero}. In practice, the performance of GSD-ST is bounded by the signal-to-noise ratio (SNR). Therefore, we can infer that the proposed GSD-ST would be beneficial when the SNR of each signal is high and the orthogonality between the signals cannot be ensured.

To harness the advantage of GSD-ST, we propose a novel technique (No-INFRA) as an application of GSD-ST to radio access networks. It strives to eliminate the effect of collisions of multiple access attempts by permitting each user to randomly select its frequency within a limited bandwidth and by employing GSD-ST for the demodulation process. To clarify our contribution, we assume that a single-path channel model is applied hereafter. The following one-to-one correspondences between GSD-ST and No-INFRA hold, and thus we use these terms interchangeably.

\begin{itemize}
    \item Number of geometric sequences ($k$) $\rightleftharpoons$ Number of signals containing independent messages.
    \item Non-orthogonally superposed $k$ geometric sequences with noise ($\mathbf{s}_w$) $\rightleftharpoons$ Sampled signal at the receiver.
    \item Initial term of the $n$-th geometric sequence ($a_n$) $\rightleftharpoons$ Multiplication of the $n$-th symbol and the channel gain between the $n$-th transmitter and the receiver.
    \item Common ratio of the $n$-th geometric sequence ($r_n$) $\rightleftharpoons$ Exponential function of the Doppler-shifted subcarrier carrying the $n$-th messages.
\end{itemize}

\subsection{Design of No-INFRA}

Let $f_n$ denote the frequency of the transmitted signal of the $n$-th transmitter. We consider that $f_n$ follows the uniform distribution, $f_n \sim \mathcal{U}(1/T,F)$, for any $n$. Here, $T$ and $F$ are the symbol duration and signal bandwidth, respectively. Each transmitter uses a single subcarrier to deliver information. For a continuous time duration $t \in [0,T]$, the baseband signal of the $n$-th transmitter can be expressed as $x_n ( e^{ j 2\pi f_n })^t$. Here, $x_n$ is a modulated symbol containing the information transmitted by the $n$-th baseband signal, where $\mathbb{E}[||x_n||^2]=1$. Then, for a discrete sampling domain $l \in \{0,1,\cdots,P-1 \}$, the discrete baseband sequence at the receiver is given by
\begin{equation}
    \mathbf{s}_w := \sum_{n=1}^k \{\beta_ne^{j\theta_n} x_n ( e^{j 2\pi \tilde{f}_n})^{l\Delta T_s}\}_{l=0}^{P-1} + \mathbf{w},
\end{equation}
where $\Delta T_s$ and $\tilde{f}_n$ are the sampling interval and the Doppler-shifted subcarrier of the $n$-th transmitter, respectively. In addition, $\beta_n$ and $\theta_n$ are the channel coefficients between the $n$-th transmitter and the receiver related to the power attenuation (owing to path-loss and shadowing) and the phase rotation (owing to the delay spread and the Doppler frequency), respectively, which can be estimated at the receiver. The channel gain and modulated symbol, i.e., $\beta_n$, $\theta_n$, and $x_n$, are integrated into the initial term of the $n$-th sequence, $a_n$. Furthermore, $\Delta T_s$ and $\tilde{f}_n$ are integrated to the common ratio of the $n$-th sequence, $r_n$. The procedure for No-INFRA is described in Algorithm 3.

\begin{algorithm}[t!]
 \caption{: Procedure for No-INFRA}\label{alg:rmleft}
 \begin{algorithmic}[1]
\State $<$\textit{transmitter side}$>$
\State Modulate the information to the modulated symbol, $x_i$.
\State Determine the frequency of signal $f_n \sim \mathcal{U}(1/T,F)$ to carry information.
\State Transmit the modulated continuous signal, $x_n ( e^{ j 2\pi f_n })^t$ for the time interval $[0, T]$.
\State $<$\textit{receiver side}$>$
\State Set $\mathbf{s}_w$ by a discrete sampling.
\State Operate Algorithm 2 with $\mathbf{s}_w$ to estimate $k$, $\mathbf{a}$, and $\mathbf{r}$.
\State Demodulate $x_n$ with the prior knowledge of $\beta_n$ and $\theta_n$, i.e., $x_n = a_n / (\beta_n e^{j2\pi \theta_n})$, for all $n$.
 \end{algorithmic}
\end{algorithm}

\section{Performance Analysis}

In this section, we evaluate the performance of the proposed No-INFRA mainly in terms of symbol error rate (SER). The result is based on $2 \times 10^5$ Monte Carlo simulation experiments. We assume the center frequency, signal bandwidth, and symbol duration to be 6 GHz, 1 MHz, and 30 $\mu$s, respectively. Furthermore, we set the sampling rate at the receiver to be equivalent to the signal bandwidth, i.e., $P=30$. The SNR of each signal is assumed to follow a normal distribution, $\mathcal{N}(\gamma_{dB}, \sigma_{dB}^2)$, in dB scale.


Additionally, we assume that the delay spread and the Doppler frequency of each signal follow uniform distributions in the ranges [0, 1 $\mu$s] and [-1 kHz, 1 kHz], respectively. Note that No-INFRA is robust to the frequency distortion caused by the Doppler effect because the users select frequencies randomly in the beginning. The robustness to the Doppler effect may be increased further if No-INFRA is combined with orthogonal time frequency space (OTFS) modulation which utilizes the delay-Doppler domain \cite{hadani2017orthogonal, hadani2018otfs}. However, it is beyond the scope of this paper and would be considered in a future study.

\begin{figure*}
        \centering
        \begin{subfigure}[b]{0.48\textwidth}
            \centering
            \includegraphics[width=\textwidth]{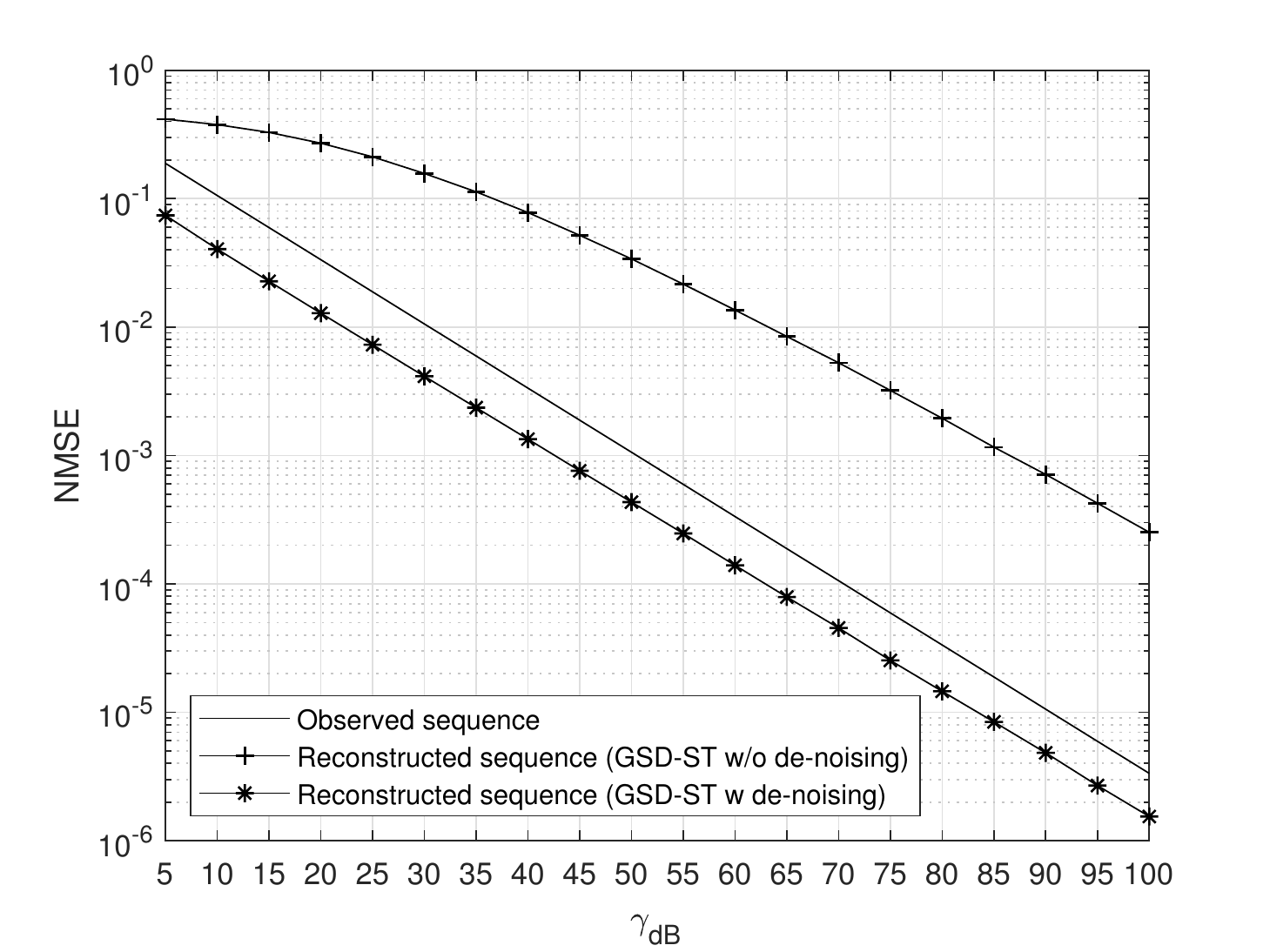}
            \caption[]%
            {{\small Effect on GSD-ST in terms of NMSE}}
            \label{DN1}
        \end{subfigure}
        \hfill
        \begin{subfigure}[b]{0.48\textwidth}
            \centering
            \includegraphics[width=\textwidth]{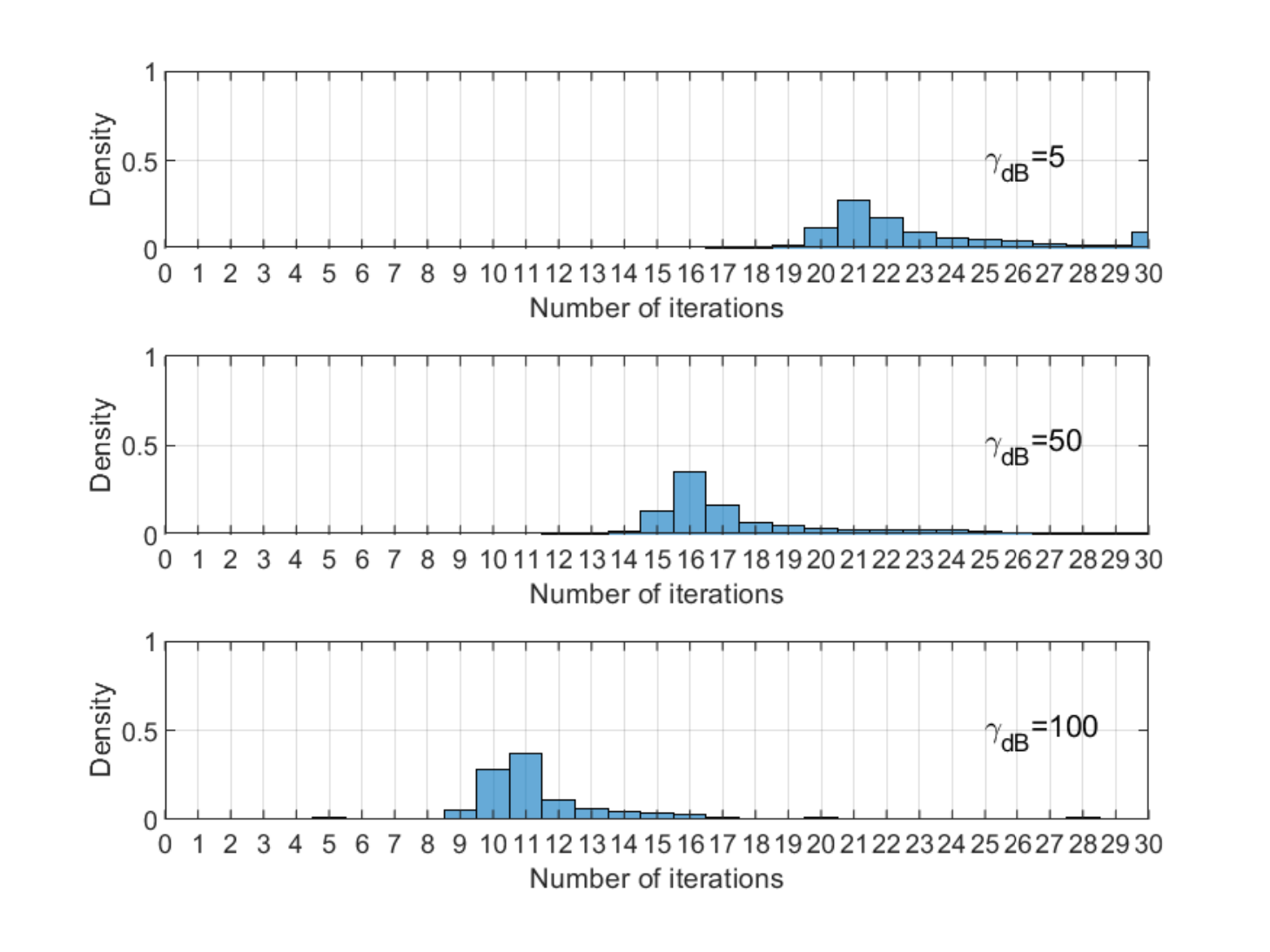}
            \caption[]%
            {{\small Convergence speed}}
                    \label{DN2}
        \end{subfigure}
        \caption{Performance of the de-noising process.
        }
        \label{DN}
    \end{figure*}

\subsection{Effect of De-noising Process}
First, we examine the impact of the de-noising process on the performance of GSD-ST. Fig. \ref{DN} shows the performance of de-noising in terms of normalized mean square error (NMSE) between the original sequence and observed/reconstructed sequences when $k$ is known. We set $\sigma_{dB}$ to zero, which implies that all the sequences undergo identical $\gamma_{dB}$. The convergence speed of de-noising is also presented to indicate its complexity.
Each transmitter selects the frequency of subcarrier through a continuous uniform distribution in the range $[33.33 \text{ kHz}, 1 \text{ MHz}]$. We set the stopping criterion ($\epsilon$) and maximum number of iterations ($I_{max}$) as $10^{-10}$ and 30, respectively.

Fig. \ref{DN1} indicates that the de-noising is an appropriate pre-processing of GSD-ST. Observe that the NMSE of the observed sequence is inverse-proportional to the SNR. As anticipated, the reconstruction of the sequence through GSD-ST incurs more errors without the de-noising. Conversely, the de-noising makes the reconstructed sequence even closer to the original one. The gap between GSD-ST with the de-noising and the observed sequence remains almost constant regardless of $\gamma_{dB}$. This indicates that the de-noising is effective in the whole range of SNR. Hence, we continue to employ the de-noising in the subsequent experiments.

The computational complexity required for de-noising is derived as $\mathcal{O}(IkP_h(P-P_h+1))$ based on the complexity of $k$-truncated SVD \cite{zhang2014randomized}. Here, $P_h=\floor{\frac{P+1}{2}}$ and $I$ is the number of iterations. The distributions of $I$ for different $\gamma_{dB}$ values are shown in Fig. \ref{DN2} to represent the complexity of de-noising and its convergence tendency. It is observed that the de-noising converges faster in the high SNR regime. When $\gamma_{dB}=5$, the de-nosing process fails to converge for 9.35$\%$ of the cases. In contrast, the de-noising is completed in a relatively short time for the high SNR, e.g., 11.33 iterations when $\gamma_{dB}=100$.

\subsection{Comparison with Conventional Scheme in Random Access}

In this subsection, the SER performance of No-INFRA is shown as functions of the SNR distribution ($\mathcal{N}(\gamma_{dB}, \sigma_{dB}^2)$), number of transmitters ($k$), and modulation order ($M$) under the assumption that $k$ is known.
QAM is adopted for modulation in our simulation.
For a performance comparison, we select orthogonal random access with successive interference cancellation (ORA+SIC) \cite{islam2016power, choi2017noma}. It consists of two steps for demodulation: the fast Fourier transform (FFT) in the time domain and SIC in the power domain. The same symbol duration and the signal bandwidth as No-INFRA yield 30 orthogonal subcarriers. However, even a few transmitters may experience interference owing to collisions that are incurred by the uncoordinated nature of the random access. For the case of collision, SIC is employed to reduce the SER.

\begin{figure}[t!]
\centering
\includegraphics[width=0.50\textwidth]{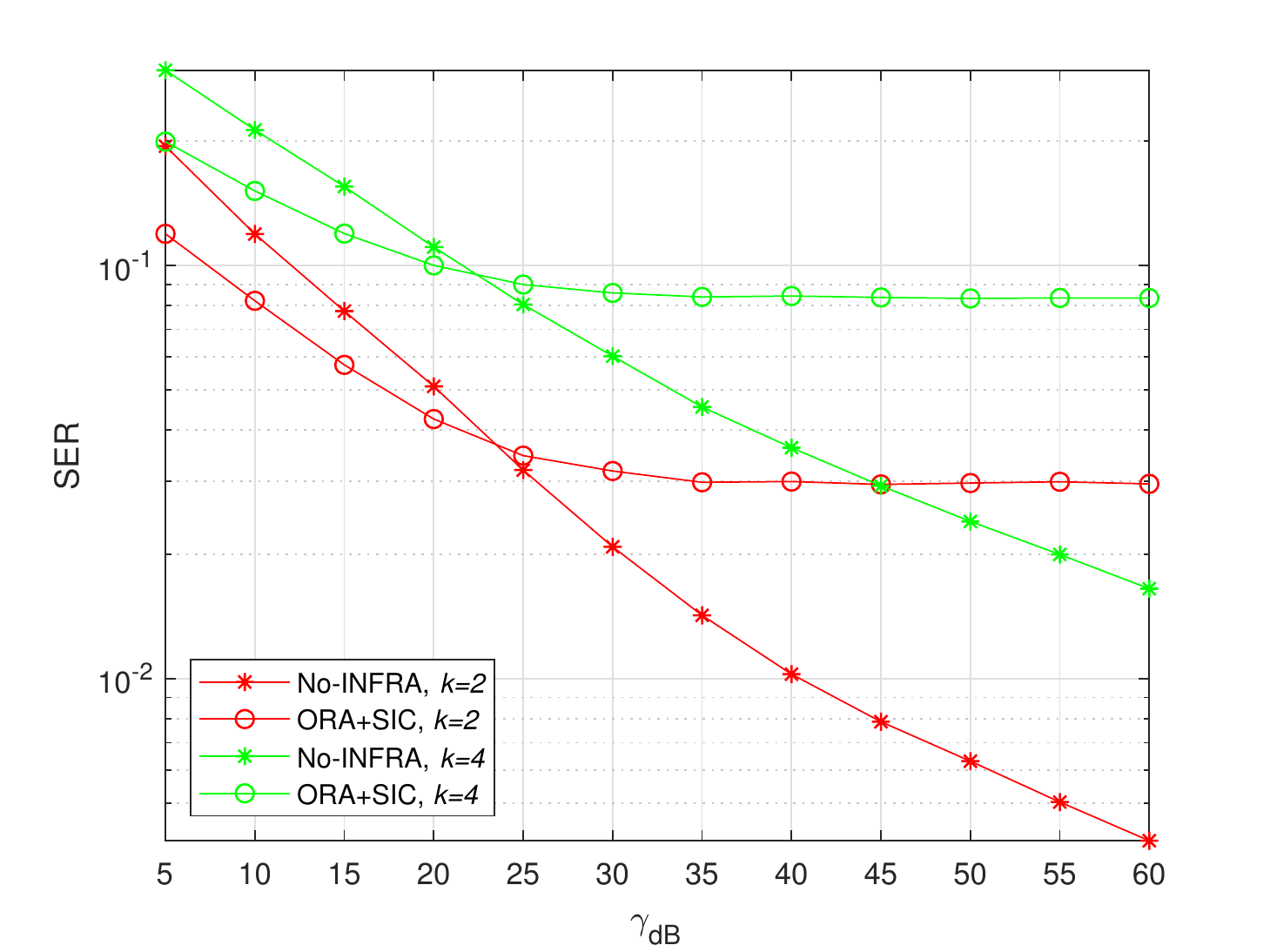}
\caption{SER according to $\gamma_{dB}$ ($M=16$ and $\sigma_{dB}=10$).}
\label{SERGAMMA}
\end{figure}

Fig. \ref{SERGAMMA} illustrates the SER of the No-INFRA and ORA+SIC schemes concerning $\gamma_{dB}$ under the setting of $M=16$ and $\sigma_{dB}=10$. This figure shows that the performance of ORA+SIC is saturated even in the high SNR regime. In the case of $k=2$, the average signal-to-interference-plus-noise ratio (SINR) is approximately 0 dB if a collision occurs. Thus, irrespective of how large $\gamma_{dB}$ is, the interference between signals remains dominant, which results in the saturation of SER performance.
This tendency is more severe for the case of $k = 4$, where the average SINR is almost -5 dB. In contrast, No-INFRA displays a remarkable SER performance over a region where interference is dominant compared to noise. That is, No-INFRA responds more strongly to a weaker noise power than ORA+SIC. Therefore, notwithstanding the poor average SINR, the SER of No-INFRA decreases linearly as $\gamma_{dB}$ increases in the log-log scale. No-INFRA starts to outperform ORA+SIC at $\gamma_{dB}=25$, and the gap widens as $\gamma_{dB}$ increases.

\begin{figure}[t!]
\centering
\includegraphics[width=0.50\textwidth]{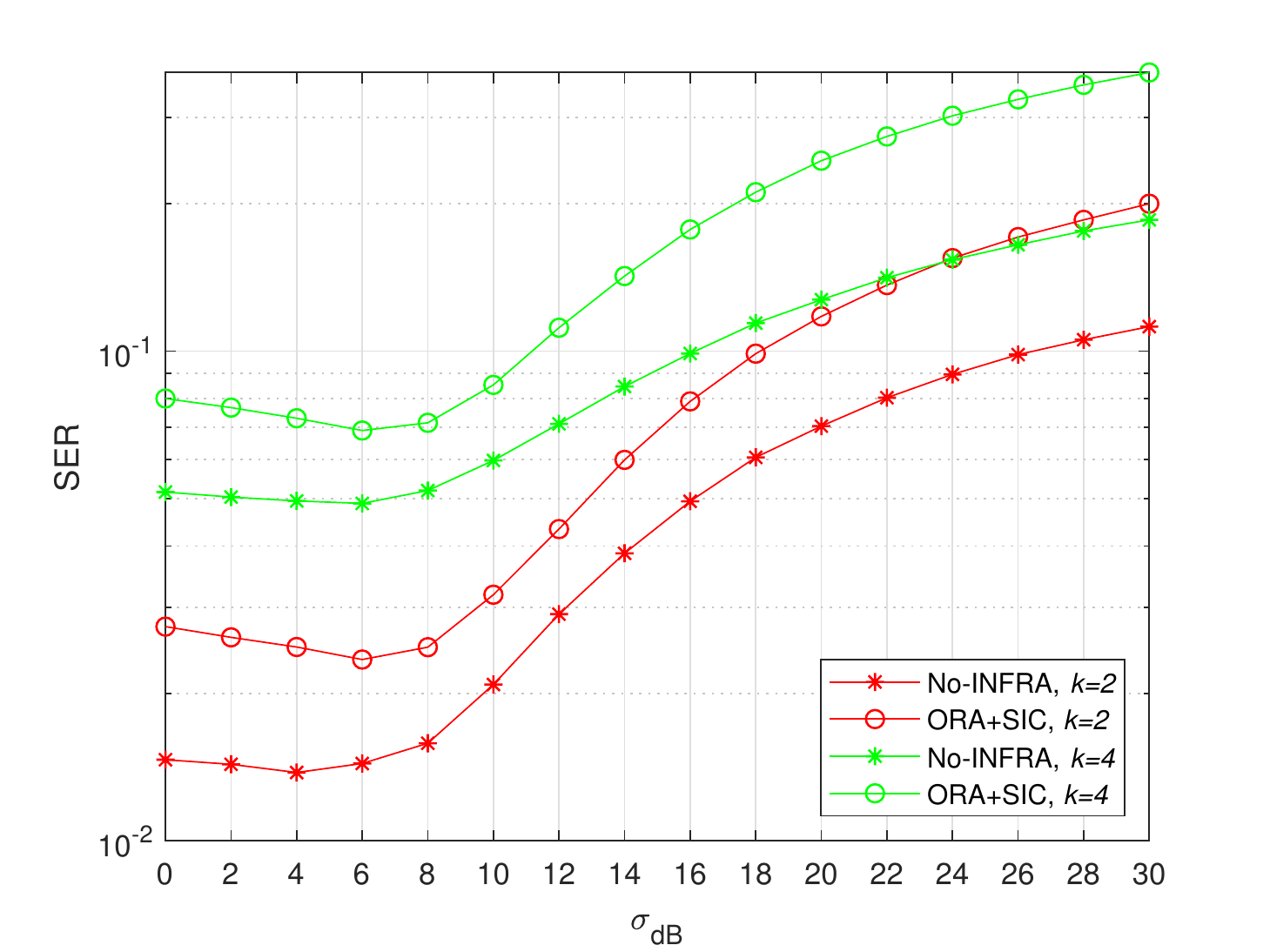}
\caption{SER according to $\sigma_{dB}$ ($M=16$ and $\gamma_{dB}=30$).}
\label{SERSIGMA}
\end{figure}

Next, the SER of No-INFRA and ORA+SIC is depicted in Fig. \ref{SERSIGMA} with respect to $\sigma_{dB} \in \{ 0, 2, \cdots ,30\}$ and with a fixed $\gamma_{dB}$. No-INFRA outperforms ORA+SIC regardless of $k$ and $\sigma_{dB}$. The SER of ORA+SIC improves as $\sigma_{dB}$ increases from zero to six because the received powers fluctuate more, which creates more suitable conditions for SIC to be effective. However, when $\sigma_{dB} > 6$, the SER of ORA+SIC deteriorates because the SINR of the weaker transmitter tends to be insufficient, whereby it demodulates only the stronger one. This phenomenon occurs in No-INFRA as well. Because the largest $k$ singular values are selected during the de-noising process, the increment in $\sigma_{dB}$ forces the small singular values to be buried in the noise.

\begin{figure}[t!]
\centering
\includegraphics[width=0.50\textwidth]{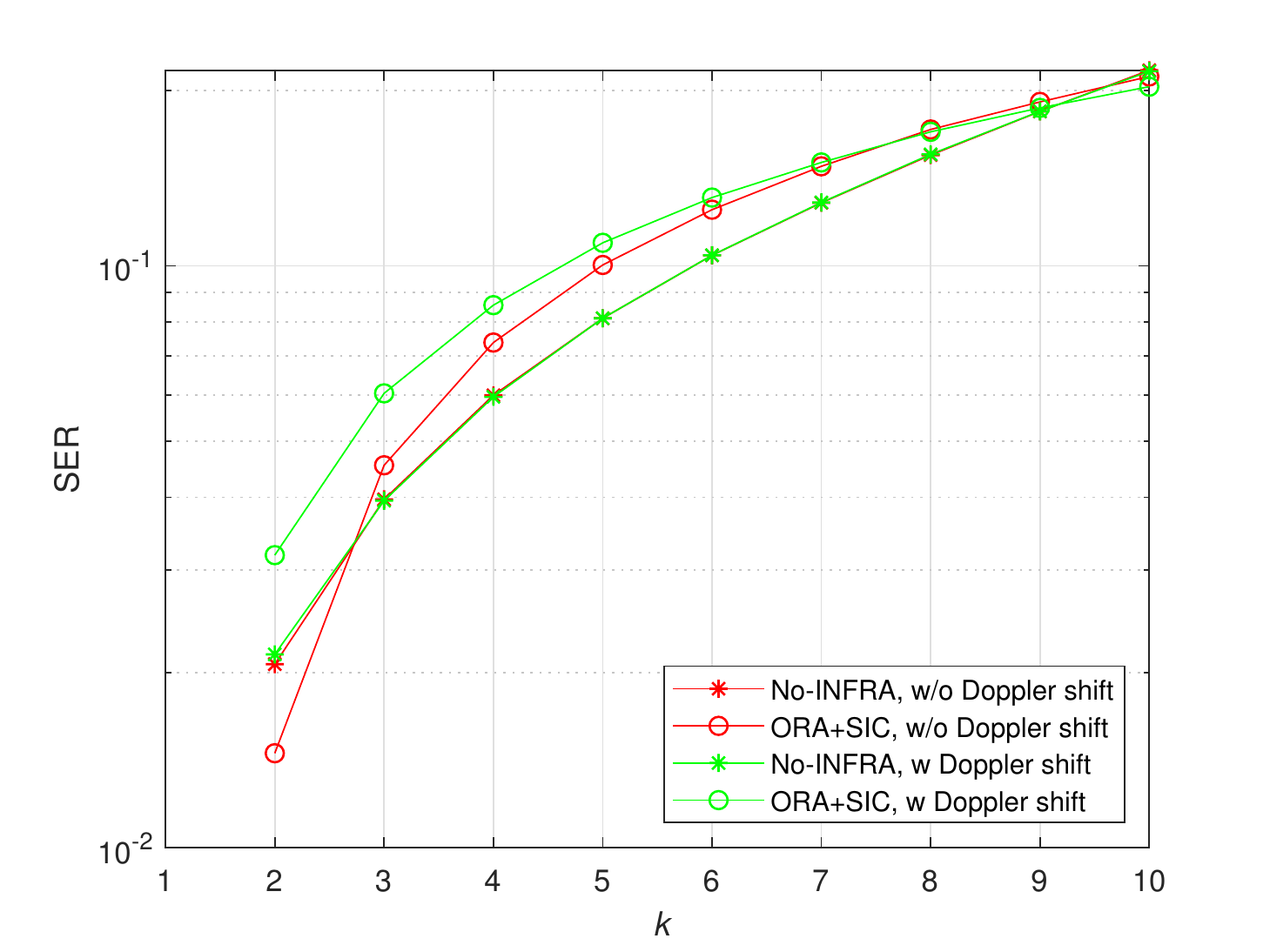}
\caption{SER according to $k$ ($M=16$, $\gamma_{dB}=30$, and $\sigma_{dB}=10$).}
\label{SERk}
\end{figure}

Fig. \ref{SERk} shows the SER of No-INFRA and ORA+SIC according to the increase in $k$. Note that the transmitters select random frequencies in No-INFRA, which implies that the Doppler shift does not influence the performance of No-INFRA. This indicates that GSD-ST can be a robust tool to combat the signal distortion by the Doppler shift. Meanwhile, ORA+SIC suffers from the orthogonality crack between the subcarriers in the presence of the Doppler shift. Therefore, No-INFRA displays a superior SER performance than ORA+SIC, particularly when the Doppler shift exists.

\begin{figure}[t!]
\centering
\includegraphics[width=0.50\textwidth]{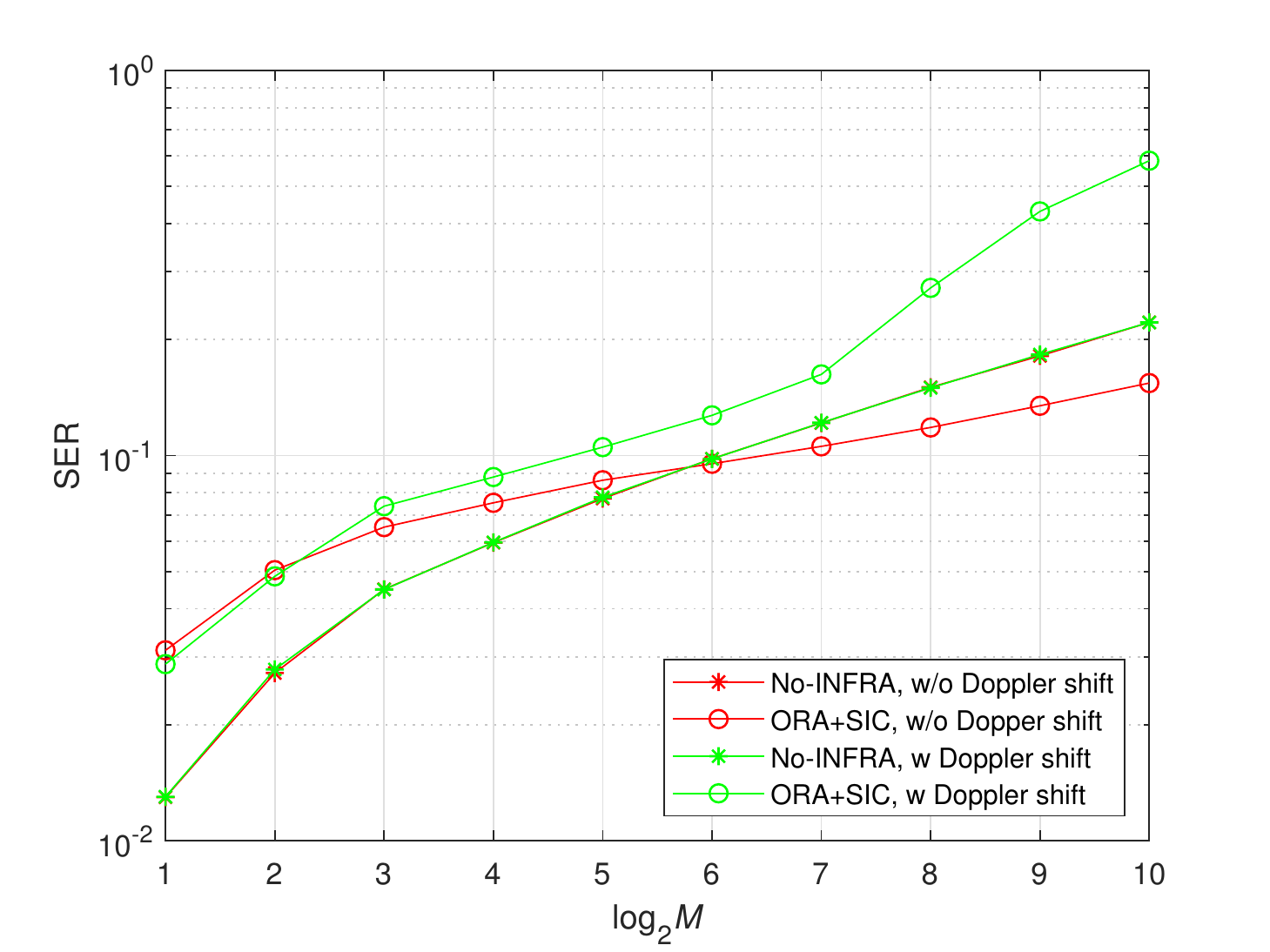}
\caption{SER according to $M$ ($k=4$, $\gamma_{dB}=30$, and $\sigma_{dB}=10$).}
\label{SERM}
\end{figure}

In Fig. \ref{SERM}, we compare the SER of No-INFRA and ORA+SIC under a variation in the modulation order, $M$. In the case of ORA+SIC, the effect of Doppler shift becomes strong at 8 QAM. For BPSK and QPSK, the demodulation is possible through the phase difference. However, from 8 QAM, the information on power difference is also required. Therefore, the performance to resolve the interference in the power domain starts to be influenced by the effect of the Doppler shift from 8 QAM. No-INFRA is devoid of this phenomenon, and therefore, outperforms ORA+SIC in the presence of the Doppler shift. However, in the static environment, ORA+SIC is more robust at high modulation orders, e.g., above 128 QAM.

To summarize, No-INFRA outperforms ORA+SIC in the scenarios where the Doppler shift and similar received powers exist. It is noteworthy that No-INFRA and ORA+SIC can be complementing techniques. Because No-INFRA is an algorithm for addressing $k$ random frequencies, it is naturally immune to the Doppler shift. Meanwhile, ORA+SIC is sensitive to the orthogonality crack owing to the Doppler shift. In addition, No-INFRA performs best when the received powers of the individual signals are similar, which is not favorable to ORA+SIC. GSD-ST utilizes the high SNR completely, whereas ORA+SIC is more robust to the noise. Therefore, No-INFRA and ORA+SIC can be alternative design options for different environments.

\subsection{Detection Rate of Number of Transmitters}

\begin{table*}[t!]
  \begin{center}
    \caption{Detection rate of $k$ based on the three similarity functions (left: $\sigma_{dB}=0$, right: $\sigma_{dB}=10$)}
    \label{T1}
\begin{tabular}
{|c||c|c||c||c|c||c||c|c|}
 \hline
 $D_{r}$ & $\gamma_{dB}=30$ & $\gamma_{dB}=60$ & $D_{d}$ & $\gamma_{dB}=30$ & $\gamma_{dB}=60$ & $D$ & $\gamma_{dB}=30$ & $\gamma_{dB}=60$ \\
 \hline
 \hline
 $k=1$ & 1.000, 1.000 & 1.000, 1.000 & $k=1$ & 1.000, 1.000 & 1.000, 1.000 & $k=1$ & 1.000, 1.000 & 1.000, 1.000 \\
 $k=2$ & 0.778, 0.627 & 0.912, 0.857 & $k=2$ & 0.833, 0.636 & 0.935, 0.878 & $k=2$ & 0.846, 0.641 & 0.947, 0.890 \\
 $k=3$ & 0.598, 0.440 & 0.797, 0.734 & $k=3$ & 0.637, 0.453 & 0.833, 0.766 & $k=3$ & 0.644, 0.459 & 0.835, 0.770 \\
 $k=4$ & 0.395, 0.261 & 0.634, 0.569 & $k=4$ & 0.436, 0.298 & 0.672, 0.599 & $k=4$ & 0.437, 0.301 & 0.6718, 0.601 \\
 \hline
\end{tabular}
  \end{center}
\end{table*}

In this subsection, we analyze the performance of the estimation of the number of transmitters, $k$. Here, we assume $k \in \{1,2,3,4 \}$.
It is equivalent to the process of estimating the number of parameters using only 30 noisy samples without additional information such as noise level. Table \ref{T1} shows the detection rates of the three similarity functions $D_r$, $D_d$, and $D$ in the environments of $\gamma_{dB} \in \{30, 60\}$ and $\sigma_{dB} \in \{0, 10\}$. In all the scenarios, the detection rate is higher at $\sigma_{dB}=0$ than at $\sigma_{dB}=10$. This indicates that a higher detection rate is obtained when the difference in the received power of each transmitter's signal is less. It is noteworthy that in the cases of $k=2$ and $\gamma_{dB}=60$, the detection rate is higher than $85\%$ for all the similarity functions. In our algorithm, the most probable $k$ is determined based on the values of the similarity function calculated over the given samples. Therefore, it is superior to techniques that require a relatively large number of samples and prior knowledge of noise levels, such as the constant false alarm rate (CFAR) algorithm \cite{rohling1983radar, tao2016segmentation}. Moreover, this method has the potential to be used as a pre-process for parametric estimation techniques such as the multiple signal classification (MUSIC) algorithm that assumes the number of parameters is known \cite{stoica1989music, zhang2018localization}.

\section{Conclusion and Future Works}

We introduced a mathematical method for decomposing non-orthogonally superposed $k$ geometric sequences, which we call GSD-ST.
Our method converts the problem of decomposing $k$ geometric sequences into root-finding of a $k$-th order polynomial equation.
We employed the concept of $k$-simplex for a formal derivation of the method and established that only $2k+1$ samples of the superposed sequence are required for the entire process of GSD-ST.

The proposed GSD-ST can be applied widely to the field of wireless communications because an equidistant sampling of a radio wave comprises a geometric sequence. We presented a new radio access scheme, namely No-INFRA, to illustrate GSD-ST's potential for addressing non-orthogonally accumulated radio signals. It enables a receiver to demodulate multiple uncoordinated access requests simultaneously. Numerical results show that No-INFRA is effective in interference-limited environments.

Considering the intrinsic similarity between radio waves and geometric sequences, we believe that the GSD-ST method can open new horizons in various research fields.
Depending on the physical domain in which the superposed radio waves are sampled, the potential of GSD-ST is broadened for sparse channel estimation that captures the features of multi-path channels such as excess delay, Doppler shift, or direction-of-arrival.
Furthermore, the GSD-ST method can be effective for mitigating the orthogonality cracks such as inter-symbol or inter-carrier interference in OFDM systems.
The abundant applicability of GSD-ST would play a crucial role in providing disruptive technologies for wireless communications.

\appendices
\section{Complex-Valued Example of GSD-ST}

Consider the following two geometric sequences:
\begin{equation}
\begin{split}
\mathbf{s}_1 & = \{ a_1 r_1^l \}_{l=0}^{P-1} = \{ (64+32j) \cdot (0.5 - 0.5j)^l \}_{l=0}^{P-1} \\ & = \{64+32j, 48 -16j, 16 -32j,  -8 -24j, -16 - 8j, -12 + 4j,  -4 + 8j,   2 + 6j, \cdots \}, \\
\mathbf{s}_2 & = \{ a_2 r_2^l \}_{l=0}^{P-1} = \{ (0.125 + 0.0625j) \cdot (2+j)^l \}_{l=0}^{P-1} \\ & = \{   0.125 + 0.0625j,   0.1875 + 0.25j,   0.125 + 0.6875j,  -0.4375 + 1.5j,  -2.375 + 2.5625j, \\ &  -7.3125 + 2.75j, -17.375 - 1.8125j, -32.9375 -21j,
  \cdots \},
\end{split}
\end{equation}
where $j=\sqrt{-1}$.
Here, we can observe only their superposition, $\mathbf{s}$, i.e.,
\begin{equation}
\begin{split}
\mathbf{s} & = \mathbf{s}_1 + \mathbf{s}_2 = \{ (64+32j) \cdot (0.5 - 0.5j)^l  + (0.125 + 0.0625j) \cdot (2+j)^l \}_{l=0}^{P-1} \\ & = \{  64.125 +32.0625j, 48.1875 -15.75j,  16.1250 -31.3125j,  -8.4375 -22.5j, \\ & -18.375 - 5.4375j, -19.3125 + 6.75j, -21.3750 + 6.1875j, -30.9375 -15j, \cdots \}.
\end{split}
\end{equation}

Let us obtain $k$, $\mathbf{a}$, and $\mathbf{r}$.

\subsubsection{Obtaining $k$}
Consider an arbitrary $\hat{k}$ as an estimate of $k$.
For $\hat{k}=2$, we consider a two-dimensional space in which we generate 2-simplexes, i.e., triangles, from the origin and consecutive values of $\mathbf{s}$. Let us create three triangles ($A_1$, $A_2$, and $A_3$) with the following coordinates:
\begin{equation}
\begin{split}
   A_1 : & [ (0,0)^{\text{T}}, (64.125 +32.0625j,48.1875 -15.75j)^{\text{T}}, (48.1875 -15.75j,16.1250 -31.3125j)^{\text{T}} ], \\
   A_2 : & [ (0,0)^{\text{T}}, (48.1875 -15.75j,16.1250 -31.3125j)^{\text{T}}, (16.1250 -31.3125j,-8.4375 -22.5j)^{\text{T}} ], \\
   A_3 : & [ (0,0)^{\text{T}}, (16.1250 -31.3125j,-8.4375 -22.5j)^{\text{T}}, (-8.4375 -22.5j,-18.375 - 5.4375j)^{\text{T}} ].
\end{split}
\end{equation}
Then, we examine whether the volumes of the triangles, $\Lambda(A_n)$, constitute a geometric sequence.
Here, $\Lambda(A_1) =  -18 +13.5j$, $\Lambda(A_2) =  -20.25 +29.25j$, and $\Lambda(A_3) =   -15.75 +54j$, which is a geometric sequence with a common ratio of $   1.5 - 0.5j$. Therefore, we verify that $\mathbf{s}$ is a superposition of \textit{two} geometric sequences $(k=2)$.
For $\hat{k}>2$, the volumes of $\hat{k}$-simplexes always constitute a sequence of zeros.

\subsubsection{Obtaining $\mathbf{a}$ and $\mathbf{r}$}

Let us extract $2k$ samples of $\mathbf{s}$ and create three vertices with the coordinates
\begin{equation}
\begin{split}
    [(64.125 +32.0625j,48.1875 -15.75j)^{\text{T}}, \\ (48.1875 -15.75j,16.1250 -31.3125j)^{\text{T}}, \\ (16.1250 -31.3125j,-8.4375 -22.5j)^{\text{T}}].
\end{split}
\end{equation}
Next, select $k$ vertices out of $k+1$ shown above. By including the origin, we can create $k+1$ $k$-simplexes in a lexicographically ordered manner. This corresponds to three triangles in this example, with the following coordinates:
\begin{equation}
\begin{split}
   B_1 : & [(0,0)^{\text{T}}, (64.125 +32.0625j,48.1875 -15.75j)^{\text{T}}, (48.1875 -15.75j,16.1250 -31.3125j)^{\text{T}}], \\
   B_2 : & [(0,0)^{\text{T}}, (64.125 +32.0625j,48.1875 -15.75j)^{\text{T}}, (16.1250 -31.3125j,-8.4375 -22.5j)^{\text{T}}], \\
   B_3 : & [(0,0)^{\text{T}}, (48.1875 -15.75j,16.1250 -31.3125j)^{\text{T}}, (16.1250 -31.3125j,-8.4375 -22.5j)^{\text{T}}].
\end{split}
\end{equation}
Let $\Lambda(B_n)$ denote the volume of the $n$-th tetrahedron.
Again, the following relationship holds by \textbf{Theorem 2} in Section \ref{combinatorial_simplex}:
\begin{equation}
    \left\{ \frac{\Lambda(B_1)}{\Lambda(B_1)}, \frac{\Lambda(B_2)}{\Lambda(B_1)}, \frac{\Lambda(B_3)}{\Lambda(B_1)}\right\} = \left\{ 1 , (r_1 + r_2), r_1 r_2 \right\}.
\end{equation}
Observe that these are the coefficients of a polynomial whose roots are $r_1$ and $r_2$. Therefore, the common ratios of the geometric sequences can be obtained by solving the quadratic equation shown below:
\begin{equation}
    x^2 - \frac{\Lambda(B_2)}{\Lambda(B_1)} x + \frac{\Lambda(B_3)}{\Lambda(B_1)} =0.
\end{equation}
Finally, once the common ratios of the sequences are obtained, we can extract the initial terms by solving a simple linear system of equations.

\section{Proof of Theorem 2}

Let $\Bar{\mathbf{\Omega}}_j \in \mathbf{C}^{(k+1) \times k}$ be the matrix form of the $j$-th union polyhedron. Similar to the matrix decomposition in \eqref{finalD}, $\Bar{\mathbf{\Omega}}_j$ can also be decomposed as follows:

\begin{gather}\label{vol_j_j+1}
    \Bar{\mathbf{\Omega}}_j =
    \mathbf{\Phi}
    \mathbf{\Sigma}_{\mathbf{a}}
    \mathbf{\Sigma}_{\mathbf{r}}^{j}
\cdot
        \begin{bmatrix}
    r_1^{0} & \cdots & r_k^0  \\
    r_1^{1} & \cdots & r_k^1 \\
    \vdots & \ddots & \vdots\\
    r_1^{k-1} & \cdots &  r_k^{k-1} \\
    r_1^{ k} & \cdots &  r_k^{ k}
\end{bmatrix}
^{\text{T}}.
    \end{gather}
Furthermore, we define $\Bar{\mathbf{\Psi}}$ as follows to express $\Bar{\mathbf{\Omega}}_j$ as $\mathbf{\Phi} \mathbf{\Sigma}_{\mathbf{a}} \mathbf{\Sigma}_{\mathbf{r}}^{j} \Bar{\mathbf{\Psi}}^\text{T}$:
\begin{equation}
\Bar{\mathbf{\Psi}} =
\begin{bmatrix}
    r_1^{0} & \cdots & r_k^0  \\
    r_1^{1} & \cdots & r_k^1 \\
    \vdots & \ddots & \vdots\\
    r_1^{k-1} & \cdots &  r_k^{k-1} \\
    r_1^{ k} & \cdots &  r_k^{ k}
\end{bmatrix}.
\end{equation}
To represent $\varkappa_j(\phi_{k,\mathbf{s}})$, let $\mathbf{Y}_l \in \{0,1\}^{(k+1) \times k}$ be the sketch matrix capturing all the rows excluding the $l$-th row. For example, $\mathbf{Y}_2 = \begin{bmatrix}
  1 & 0 & 0 & 0 \\
  0 & 0 & 1 & 0\\
  0 & 0 & 0 & 1
\end{bmatrix}^{\text{T}}$ when $k=3$. The volume of $\varkappa_j(\phi_{k,\mathbf{s}})[l]$ is given by
\begin{equation}
\Lambda(\varkappa_j(\phi_{k,\mathbf{s}})[l]) = \frac{1}{k!} \cdot \det(\mathbf{\Bar{\mathbf{\Omega}}}_j\mathbf{Y}_l).
\end{equation}
Then, \eqref{vj} can be rewritten as follows:

\begin{equation}
        \mathbf{v}_j(\phi_{k,\mathbf{s}}) = \{ \frac{  \det(\mathbf{\Bar{\mathbf{\Omega}}}_j \mathbf{Y}_0 )}{\det(\mathbf{\Bar{\mathbf{\Omega}}}_j\mathbf{Y}_0 )},  \cdots, \frac{  \det(\mathbf{\Bar{\mathbf{\Omega}}}_j \mathbf{Y}_k )}{\det(\mathbf{\Bar{\mathbf{\Omega}}}_j\mathbf{Y}_0 )} \} = \{ \frac{  \det(\Bar{\mathbf{\Psi}}^{\text{T}} \mathbf{Y}_0 )}{\det(\Bar{\mathbf{\Psi}}^{\text{T}}\mathbf{Y}_0 )}, \cdots, \frac{\det(\Bar{\mathbf{\Psi}}^{\text{T}} \mathbf{Y}_k )}{\det(\Bar{\mathbf{\Psi}}^{\text{T}} \mathbf{Y}_0)} \}.
\end{equation}
For simplicity, let $\mathbf{\Psi}_l$ denote $\Bar{\mathbf{\Psi}}^{\text{T}}\mathbf{Y}_l$. This implies that $\mathbf{\Psi}_l$ is the transpose of the matrix excluding the $(k-l)$-th row of $\Bar{\mathbf{\Psi}}$.
Then, we can simplify $\mathbf{v}_j(\phi_{k,\mathbf{s}})$ as follows:
\begin{equation}\label{v_reform}
\mathbf{v}_j(\phi_{k,\mathbf{s}}) =
 \{ \frac{  \det({\mathbf{\Psi}}_0 )}{\det({\mathbf{\Psi}}_0 )}, \cdots, \frac{\det({\mathbf{\Psi}}_k )}{\det({\mathbf{\Psi}}_0)} \}.
\end{equation}
Let us have a closer look at $\det(\mathbf{\Psi}_0)$.
\begin{equation}\label{Psi_0}
\det({\mathbf{\Psi}}_0) = \det
\begin{bmatrix}
r_1^0 & r_1^1 & \cdots & r_1^{k-1} \\
r_2^0 & r_2^1 & \cdots & r_2^{k-1} \\
\vdots & \vdots & \ddots & \vdots \\
r_{k-1}^0 & r_{k-1}^1 & \cdots & r_{k-1}^{k-1} \\
r_k^0 & r_k^1 & \cdots & r_k^{k-1}\end{bmatrix}.
\end{equation}
We employ the technique of variable substitution to manipulate $\det(\mathbf{\Psi}_0)$. Replace $r_k$ in $\mathbf{\Psi}_0$ with the variable $x$, and generate the following polynomial $p_0(x)$:

\begin{equation}\label{p_replace}
p_0(x) = \det
\begin{bmatrix}
r_1^0 & r_1^1 & \cdots & r_1^{k-1} \\
r_2^0 & r_2^1 & \cdots & r_2^{k-1} \\
\vdots & \vdots & \ddots & \vdots \\
r_{k-1}^0 & r_{k-1}^1 & \cdots & r_{k-1}^{k-1} \\
x^0 & x^1 & \cdots & x^{k-1} \end{bmatrix}.
\end{equation}
The roots of $p_0(x)$ are given by $x \in \{ r_1, \cdots r_{k-1}\}$. This is because the $n$-th row of $\mathbf{\Psi}_0$ is represented only by $r_n$ in \eqref{p_replace}. Thus, if $x$ is replaced by $r_n$ $(n \neq k)$, $p_0(x)$ becomes zero. It implies that $(x-r_1), \cdots, (x-r_{k-1})$ are the factors of $p_0(x)$. Consequently, we can rewrite \eqref{p_replace} as follows:
\begin{equation}\label{q_replace}
p_0(x) = q(x^0)\cdot(x-r_1)\cdots(x-r_{k-1}),
\end{equation}
where $q(x^0)$ is a coefficient of $x^{k-1}$.
Considering the rule for calculating the determinant, $q(x^0)$ is equal to the determinant of $\mathbf{\Psi}_0'$, which is presented below:
\begin{equation}\label{Psi_replace}
{\mathbf{\Psi}}_0' =
\begin{bmatrix}
r_1^0 & r_1^1 & \cdots & r_1^{k-2} \\
r_2^0 & r_2^1 & \cdots & r_2^{k-2} \\
\vdots & \vdots & \ddots & \vdots \\
r_{k-1}^0 & r_{k-1}^1 & \cdots & r_{k-1}^{k-2}
\end{bmatrix}.
\end{equation}
Insert $\det (\mathbf{\Psi}_0')$ instead of $q(x^0)$ in \eqref{q_replace}, and replace $x$ with $r_k$. Then, \eqref{Psi_0} can be rewritten as follows:
\begin{equation}
\det(\mathbf{\Psi}_0) = \det(\mathbf{\Psi}_0')
\cdot(r_k-r_1)\cdots(r_k-r_{k-1}).
\end{equation}
We can repeat the above process until $\mathbf{\Psi}_0'$ is equal to $r_1^0$, i.e., one. As a result,
\begin{equation} \begin{split}
\det(\mathbf{\Psi}_0) & = 1\cdot (r_2-r_1)  \cdots \big( (r_{k-1}-r_1)\cdots(r_{k-1}-r_{k-2}) \big)\cdot \big( (r_k-r_1)\cdots(r_k-r_{k-1}) \big)
\\
 & = \prod_{1 \leq n < m \leq k}(r_m - r_n).
\end{split}
\end{equation}
Similarly, we can construct $p_1(x)$ based on $\mathbf{\Psi}_1$ as follows:
\begin{equation}\label{q_replace1}
p_1(x) = q(x^1)\cdot(x-r_1)\cdots(x-r_{k-1}),
\end{equation}
where $q(x^1) = \prod_{1 \leq n < m \leq k-1}(r_m - r_n)\cdot (x+\sum_{n=1}^{k-1}r_n)$. By replacing $x$ in \eqref{q_replace1} with $r_k$, we obtain
\begin{equation}
\det(\mathbf{\Psi}_1) = (\sum_{n=1}^k r_n)\cdot \prod_{1 \leq n < m \leq k}(r_m - r_n) = (\sum_{n=1}^k r_n)\cdot \det(\mathbf{\Psi}_0).
\end{equation}
Repeat the above process over all possible $l$. Then, we can determine the following form of $\det(\mathbf{\Psi}_l)$:

\begin{equation}
    \det(\mathbf{\Psi}_l) = \prod_{1 \leq n < m \leq k}(r_m - r_n) \sum_{1 \leq i_1 < i_2 < \dots < i_l \leq k}   \big( \prod_{n=1}^l r_{i_n} \big).
\end{equation}
Recalling \eqref{v_reform}, $\mathbf{v}_j(\phi_{k,\mathbf{s}})$ can be obtained as follows:
\begin{equation}
\mathbf{v}_j(\phi_{k,\mathbf{s}})
=\{ 1 , \cdots , \sum_{1 \leq i_1 < i_2 < \dots < i_l \leq k}   \big( \prod_{n=1}^l r_{i_n} \big), \cdots , \prod_{n=1}^{k} r_n \}.
\end{equation}
Therefore, all possible $\mathbf{v}_j(\phi_{k,\mathbf{s}})$ are identical to each other.

\ifCLASSOPTIONcaptionsoff
  \newpage
\fi

\bibliographystyle{IEEEtran}
\bibliography{references.bib}

\end{document}